%% file: main.tex
\documentclass[10pt, conference, letterpaper]{IEEEtran}
\IEEEoverridecommandlockouts
\usepackage{cite}
\usepackage{amsmath,amssymb,amsfonts}
\usepackage{graphicx}
\usepackage{textcomp}
\usepackage{xcolor}
\usepackage{dblfloatfix}  
\usepackage{algorithm}
\usepackage[noend]{algorithmic}
\usepackage{graphicx}
\usepackage{amsmath}
\usepackage{amsthm}
\usepackage{multirow}
\usepackage{caption}
\usepackage{subcaption}
\usepackage{wrapfig}
\usepackage{dsfont}
\usepackage{footmisc}
\usepackage{mleftright}
\usepackage{url}

\mleftright
\usepackage{float}
\newtheorem{lemma}{Lemma}
\newtheorem{theorem}{Theorem}
\newtheorem{definition}{Definition}

\newcounter{relctr} 
\everydisplay\expandafter{\the\everydisplay\setcounter{relctr}{0}}

\newcommand\labrel[2]{%
  \begingroup
    \refstepcounter{relctr}%
    \stackrel{\textnormal{(\alph{relctr})}}{\mathstrut{#1}}%
    \originallabel{#2}%
  \endgroup
}
\AtBeginDocument{\let\originallabel\label} 

\def\BibTeX{{\rm B\kern-.05em{\sc i\kern-.025em b}\kern-.08em
    T\kern-.1667em\lower.7ex\hbox{E}\kern-.125emX}}
\begin{document}

\newcommand{\mg}[1]{{\color{blue} MG: #1}}
\newcommand{\yc}{\color{red}}
\newcommand{\ViS}{\mathsf{ViS}}

\title{VR Viewport Pose Model for Quantifying and Exploiting Frame Correlations}

\author{Ying Chen\IEEEauthorrefmark{1}, Hojung Kwon\IEEEauthorrefmark{1}, 
Hazer Inaltekin\IEEEauthorrefmark{2}, Maria Gorlatova\IEEEauthorrefmark{1} \\
\IEEEauthorrefmark{1}Duke University, Durham, NC, \IEEEauthorrefmark{2}Macquarie University, North Ryde, NSW, Australia\\
\IEEEauthorrefmark{1}\{ying.chen151, hojung.kwon, maria.gorlatova\}@duke.edu, \IEEEauthorrefmark{2}hazer.inaltekin@mq.edu.au}

\maketitle

\begin{abstract} 
The importance of the dynamics of the \emph{viewport pose}, i.e., the location and the orientation of users' points of view, for virtual reality (VR) experiences calls for the development of VR viewport pose models. In this paper, informed by our experimental measurements of viewport trajectories 
across 3 different types of VR interfaces, we first develop a statistical model of viewport poses in VR environments.  
Based on the developed model, we examine the correlations between pixels in VR frames that correspond to different viewport poses, and obtain an analytical expression for the visibility similarity~(ViS) of the pixels across different VR frames. 
We then propose a lightweight ViS-based ALG-ViS algorithm that adaptively splits VR frames into the background and the foreground, reusing the background across different frames. 
Our implementation of ALG-ViS in two Oculus Quest 2 rendering systems demonstrates ALG-ViS running in real time, supporting the full VR frame rate, and outperforming baselines on measures of frame quality and bandwidth consumption.

\end{abstract}

\section{Introduction}
\label{sub:introduction} 
\input{Introduction.tex}

\section{Related Work}
\label{sec: related Work}
\input{RelatedWork.tex}

\section{VR Viewport Pose Model}
\label{sec:pose}

We introduce our collected dataset in \S\ref{dataset}, describe our orientation model in \S\ref{sub:orientation} and our position model in \S\ref{modifiedrwp}, and model the correlation between them in \S\ref{sec:correlation}. 

\input{Six_DoF_Camera_Pose_Model.tex}

\subsection{Position Model}
\label{modifiedrwp}

\input{positionModel.tex}

\section{Visibility Similarity}
\label{sec:visibilitysimilarity}
\input{VisibilitySimilarity}

\section{Evaluation}
\label{sub:evaluation}
\input{Evaluation.tex}


\section{Conclusion}
\label{sub:conclusion}
\input{Conclusion.tex}

\section*{Acknowledgments}

 This work is supported in part by NSF grants CSR-1903136, CNS-1908051, and CAREER-2046072, and by an IBM Faculty Award.

\bibliographystyle{IEEEtran}
\bibliography{IEEEabrv,Sig, MGbib}

\appendix
\input{appendix.tex}

\end{document}

%% file: Introduction.tex
    Virtual reality (VR), which immerses users into computer-generated virtual environments~\cite{lavalleIntro}, 
        has been showing promise in many applications including gaming, education, and healthcare \cite{gpp}. 
        VR is expected to boost global GDP 
        by \$450~billion by 2030~\cite{PWC2019Seeing}. 
        High expectation for VR, coupled with its known resource-hungry nature~\cite{Cuervo18},   
        spurred a wide range of recent research that optimizes VR systems to reduce their communication and computing resource  consumption~\cite{VRLTE18,TiledStreaming,Liu2018Cutting,YangLi2019,YLi18,MengJiayi20}.

A particular feature of VR is the tight coupling of user's actions and the generated frames. In traditional visual media, the frames that are shown to the users are 
fixed. By contrast, in VR, to allow the users to independently explore virtual worlds, each frame is generated for the specific point of view of the user at a given time, i.e., for the specific \emph{viewport pose}, 
namely the $x$, $y$, and $z$ coordinates, and polar and azimuth orientation angles $\theta$ and $\phi$, of user's VR headset or another interface to the virtual world 
(mobile phone \cite{Firefly,YangLi2019,VRLTE18}; computer monitor~\cite{2019desktop}). Hence, the correlations between different VR frames, and the performance of 
approaches that exploit them 
to reduce resource consumption in VR~\cite{YangLi2019,YLi18,MengJiayi20}, are \emph{intimately tied to the dynamics of user behavior within the VR experience}. We examine and exploit this phenomenon in this work.

    First, we develop a statistical model of users' VR viewport pose, comprised of the models of pose components, orientation and position. 
        To develop this model, we collected a dataset of VR viewport trajectories in 3 VR games 
        and across 3 different types of VR user interfaces, 
            with over 5.5 hours of user data in total. 
       
    To characterize the correlation of viewport orientations between VR frames that are $\Delta t$ seconds apart, we obtain models of the change of azimuth and polar angles over $\Delta t$. 
    To characterize the displacement between VR viewport positions that are $\Delta t$ apart, 
        we propose a \emph{modified random waypoint model (RWP) with random pause times} (`paused-MRWP'). 
    We demonstrate a close fit of the developed 
    VR viewport pose 
    model to the experimental data. 
            \emph{To the best of our knowledge, this is the first statistical model of viewport pose in VR.}
 
  \begin{wrapfigure}{R}{0.2\textwidth}
 \vspace{-0cm}
  \centering
  \includegraphics[width=0.8\linewidth]{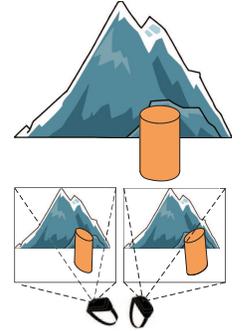}
   \vspace{-0.2cm}
  \caption{VR frames generated for similar VR device viewport poses.
  \label{fig:PixelReusePotential}}
  \vspace{-0.5cm}
\end{wrapfigure}

    Next, we apply the developed pose model to quantify the similarity of pixels across VR frames. 
        For similar poses, the VR frames 
        are highly redundant, as shown in Fig.~\ref{fig:PixelReusePotential}. 
        It is thus possible to reduce 
        resource consumption 
            by rendering a set of `reference' frames, and generating other, 
            `novel', frames 
            by rendering only a portion of the frame 
            while generating the rest by reusing  the reference frame via view projection~\cite{YangLi2019, Warp}.
    In this paper we derive analytical expressions for the visibility similarity (ViS) of pixels across different VR frames, relating  
            the poses of the reference and novel frames through the developed viewport pose model, and accounting for the 
    misalignments of the fields of view (FoVs) 
    and 
    the VR contents-to-viewport distance differences between the novel and the reference frames. 
    We verify our analysis via Unity~3D~\cite{Unity20} game engine-based simulations. 
    Finally, we exploit
    the formulated ViS to adaptively divide VR frame contents into \emph{background and foreground}, in order to render the foreground for the novel frames, while reusing the background. 
        Separate treatment of background and foreground in VR frame generation has been considered in multiple lines of work~\cite{Firefly, Furion17,MengJiayi20,YLi18,VRPredictiveSchedule}, which use heuristics for this separation. 
        In this work, we propose a lightweight algorithm, ALG-ViS, 
        that uses the analytical ViS to 
        adaptively determine the distance threshold beyond which the contents
        are treated as background.
We incorporate the developed ALG-ViS in two  rendering systems based on Oculus Quest~2 (also known as Meta Quest~2), one on-device and one supported by edge computing. In both systems, \mbox{ALG-ViS} runs in real time, supporting the full VR frame rate, and outperforming a set of baselines on measures of frame quality and resource consumption.

To summarize, the main contributions of this paper are: (i) the first statistical model of viewport pose in VR, (ii) the analysis of the visibility similarity between different VR frames, and (iii) the analytically grounded algorithm for determining which contents to reuse across different frames.  We make  the VR viewport pose dataset and our  implementation codes publicly available via GitHub.\footnote{https://github.com/VRViewportPose/VRViewportPose. \label{footnote:gitHubLink}}

The rest of this paper is organized as follows. 
We review the related work in~\S\ref{sec: related Work}, 
propose the 
viewport pose model in~\S\ref{sec:pose}, and 
analyze the ViS and propose the ALG-ViS in~\S\ref{sec:visibilitysimilarity}. We present the evaluation in~\S\ref{sub:evaluation} and conclude the paper in~\S\ref{sub:conclusion}. 

%% file: RelatedWork.tex
\par \noindent
\textbf{Device pose modeling}: 
VR frame generation requires information about the  pose (position and orientation) of user's point of view. The vast body of work that has, over the years, 
modeled 
human mobility in many different applications~\cite{C.Bettstetter03, S.Ioannidis06,I.Rhee11} 
focused on human positions but not orientations. Orientations of handheld mobile devices are starting to be modeled in context of visible light communications~\cite{Y.S.Eroglu19,M.D.Soltani19}. \emph{We are unaware of existing statistical models of users' viewport pose in VR}. 
The position component of our developed
model builds on the modified RWP proposed in~\cite{XLin13}, and one of the orientation components is related to the observations previously made in~\cite{VSitzmann2018}. The comprehensive 
model we propose significantly modifies and extends these approaches. 

\par \noindent
\textbf{Predicting VR viewport pose}: {Recently, several approaches that predict pose or its components in VR systems have been developed~\cite{Flare,kNN,Firefly,6DoFVRPrediction,PredictiveStreaming,LiveObj}.}  
Unfortunately highly immersive VR experiences are known to be negatively affected by errors in the pose prediction~\cite{YangLi2019}. Our statistical approach can be seen as making decisions based on the \emph{distribution of viewport poses}, rather than the specific predicted pose. Our evaluation demonstrates that this approach improves image quality and bandwidth variability over prediction-based approaches.

\par \noindent 
\textbf{Exploiting redundancy across VR frames}: Multiple methods 
for reducing the required  bandwidth and transmission latency in VR have 
been developed~\cite{VRLTE18,TiledStreaming,Liu2018Cutting,YangLi2019,YLi18,MengJiayi20}. 
In a rich body of work~\cite{Firefly, Furion17,MengJiayi20,YLi18,VRPredictiveSchedule}, a VR frame is classified into \emph{background} 
that is relatively static across VR frames 
and \emph{foreground} 
that is less similar from one frame to the next. The background can be rendered on the edge and prefetched by the VR device, while the foreground is rendered on the mobile device \cite{Firefly,Furion17,VRPredictiveSchedule}; the background can also be reused across multiple frames~\cite{MengJiayi20,YLi18}. 
These studies 
use heuristics to separate the background and the foreground. 
Complementing this work, we derive an analytical expression for the inter-frame pixel similarity, which we use to  
 split the background and the foreground \emph{adaptively}, via a lightweight algorithm that can run on-device or on the edge server.
 Our evaluation demonstrates that our 
 approach improves the VR image quality while consuming fewer 
 resources.

%% file: Six_DoF_Camera_Pose_Model.tex
\subsection{Collected Dataset}
\label{dataset}

To complement existing 
datasets of users' head orientation in 360$^\circ$ videos~\cite{VSitzmann2018,dataset2,360dataset} 
and a small-scale single-interface dataset of users' head pose in untethered VR~\cite{VRdataset}, 
we collected a dataset of users' viewport pose in 3 different VR games listed in Table~\ref{tab:vrgame}, across 3 different common VR user interface types. 
Specifically, we examined: (i) VR experienced through a VR headset and controlled through user head rotation and a VR controller (``headset VR"), (ii) ``desktop VR"~\cite{2019desktop}, experienced through the user's desktop monitor and controlled through desktop's mouse and keyboard, 
and (iii) VR experienced through a mobile phone, controlled via moving the phone and tapping on it~\cite{VIVO}. 
{ Our institutional review board (IRB)-approved data collection,} conducted under \mbox{COVID-19} restrictions, 
involved 
remote desktop~VR and phone-based VR data collection via apps that we distributed to remote users, and a small number of socially distanced in-lab experiments for headset and phone-based VR. In total, we recorded experiences of 5 users with headset and phone-based VR, and 20 users with desktop VR.  
{For desktop and phone-based VR, each user explored the 3 VR games for 2--5 minutes (per game). }
For headset VR, the users explored each game for 2 minutes to avoid simulator sickness. 
{ Additional data collection protocol details and the dataset are provided via GitHub.\footref{footnote:gitHubLink} }

 \begin{table}
\caption{Main characteristics of VR games. 
} 
\vspace{-0.2cm}
\centering
\label{tab:vrgame}
\begin{tabular}{ccc c c c}
\hline
VR game&Number of triangles&Number of vertices
\vspace{-0.35cm}
\\ \\
\hline
Viking village (VK) \cite{VikingVillage15} &2,400 K& 1,600 K \\
Lite \cite{Lite17} & 65.7 K &52.4 K \\
Office \cite{Office20} & 207.6 K &143.7 K\\
\hline
\end{tabular}
\vspace{-0.5cm}
\end{table}

\subsection{Orientation Model}
\label{sub:orientation}
\begin{wrapfigure}{R}{0.175\textwidth}
  \centering
     \vspace{-1.3cm}
  \includegraphics[width=0.165\textwidth]{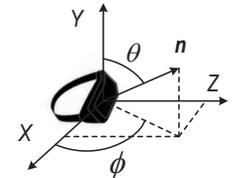}
   \vspace{-0.3cm}
    \caption{VR viewport  orientation representation.
    }
  \label{fig:deviceorientation}
  \vspace{-0.2cm}
\end{wrapfigure}

A VR viewport 
is depicted in Fig.~\ref{fig:deviceorientation}, where $\bf n$ is the unit vector along the optical axis of the camera. 
We first introduce the representation for the viewport orientation using variables related to ${\bf n}$, followed by the definition of the statistical viewport orientation model.

\begin{definition}[Viewport orientation representation]
The viewport orientation is the tuple $(\theta, \phi)$, where polar angle $\theta \in \left[0, \pi\right]$ is the angle between $\boldsymbol{n}$ and the positive direction of the $Y$-axis, and the azimuth angle $\phi \in \left[-\pi, \pi\right)$ is the angle between the projection of $\boldsymbol{n}$ in the $XZ$-plane and the positive direction of the $X$-axis in the Earth coordinates $XYZ$. 
$\theta$ and $\phi$ characterize how users look vertically and horizontally.
\end{definition}

\begin{definition} [Viewport orientation model] 
The viewport orientation model is the tuple $\left(p_\theta(\theta), p_{\Delta \theta}(\Delta \theta), p_{\Delta \phi}(\Delta \phi) \right)$, where $p_\theta(\theta)$, $p_{\Delta \theta}(\Delta \theta)$, and ${p_{\Delta\phi} }(\Delta\phi )$ are the probability density functions (PDFs) of the polar angle $\theta$,  the polar angle change $\Delta \theta \in \left[-\pi,\pi\right]$ over the time interval $\Delta t$, and  the azimuth angle change $\Delta \phi \in \left[-\pi,\pi\right]$ over $\Delta t$. In this paper we assume that ${p_\theta }(\theta )$, ${p_{\Delta\theta} }(\Delta\theta)$, and ${p_{\Delta\phi} }(\Delta\phi )$ are independent. 
\end{definition}

 $\Delta \theta$ is given by $\Delta \theta=\theta_ {\rm{nov}}-\theta_{\rm{ref}}$, where $\theta_ {\rm{ref}}$ and $\theta_{\rm{nov}}$ are the polar angles of the reference and novel frames taken
 $\Delta t$ seconds apart. $\Delta \phi$ is calculated as $\Delta \phi  =  - \pi  + \bmod \left( {{\phi _{\rm{nov}}} - {\phi _{\rm{ref}}} + \pi ,2\pi } \right)$, where $\bmod \left( {a,b} \right) = a - b\left\lfloor {\frac{a}{b}} \right\rfloor$ and $\left\lfloor  \cdot  \right\rfloor$ is the floor function, and $\phi_ {\rm{ref}}$ and $\phi_{\rm{nov}}$ are the azimuth angles of the reference and novel frames taken $\Delta t$ seconds apart. 

We examine 
azimuth angle \emph{change} $\Delta \phi$ rather than $\phi$ itself because 
$\phi$ can be assumed to be uniformly distributed when VR contents are  
scattered along different longitudes in the VR systems and there are no viewing preferences. Hence, the distribution of $\phi$ does not provide information about the correlation between different VR frames.

\begin{table*}[htbp]
\centering
  \caption{Mean and scale of the Laplace distribution that has the smallest SSE with the experimental measurements of polar angle $\theta$, for different interface types and VR games.}
  \label{our_polar_model}
  \begin{tabular}{c|ccc |ccc |ccc}
    \hline
Interface type &\multicolumn{3}{c|}{Desktop VR} &\multicolumn{3}{c|}{Headset VR}&\multicolumn{3}{c}{Phone-based VR}\\
\hline
 VR game&VK&Lite& Office& VK & Lite & Office& VK & Lite & Office\\
\hline
Mean & 90.575 & 90.037 & 89.979 & 92.331&89.654 &88.999& 90.896&90.548 &89.944\\
Scale & 7.356 & 6.057 &6.204& 4.319 &3.454 &3.646& 6.797&6.195 &6.856\\
\hline
  \end{tabular}
  \vspace{-0.2cm}
\end{table*}

\begin{table*}[t]
\centering
\caption{Scale of best-fitting Laplace distributions for polar angle change $\Delta \theta$, in Lite for different interface types.}
\label{Table Polar Angle Change}
\vspace{-0.1cm}
\begin{tabular}{c| c|c c c c c c c c cc cc}
\hline
\multicolumn{2}{c|}{$\Delta t$ (1/60 s)}  & 1& 5&10&15&20& 25 &30& 100 &200&600 &5000 \\
\hline
\multirow{3}{3em}{Interface type} &Desktop VR & 0.0748& 0.305 & 0.571& 0.826 & 1.049& 1.276 & 1.480& 3.407 & 4.984 & 8.005 &9.926 \\
& Headset VR & 0.0841& 0.450 & 0.825 & 1.149 &1.438 & 1.699 & 1.920 & 3.516 & 4.510 & 5.816 & 6.186
\\
& Phone-based VR & 0.0740& 0.332 &0.664 &0.727 & 0.880 & 1.013 & 1.141 &2.589 & 3.495 &4.792 & 12.648\\
\hline
\end{tabular}
\vspace{-0.1cm}
\end{table*}

\subsubsection{Distribution fit}

\begin{figure*}
\begin{minipage}{0.24\linewidth}
  \centering
   \vspace{-0.2cm}
  \includegraphics[width=\textwidth]{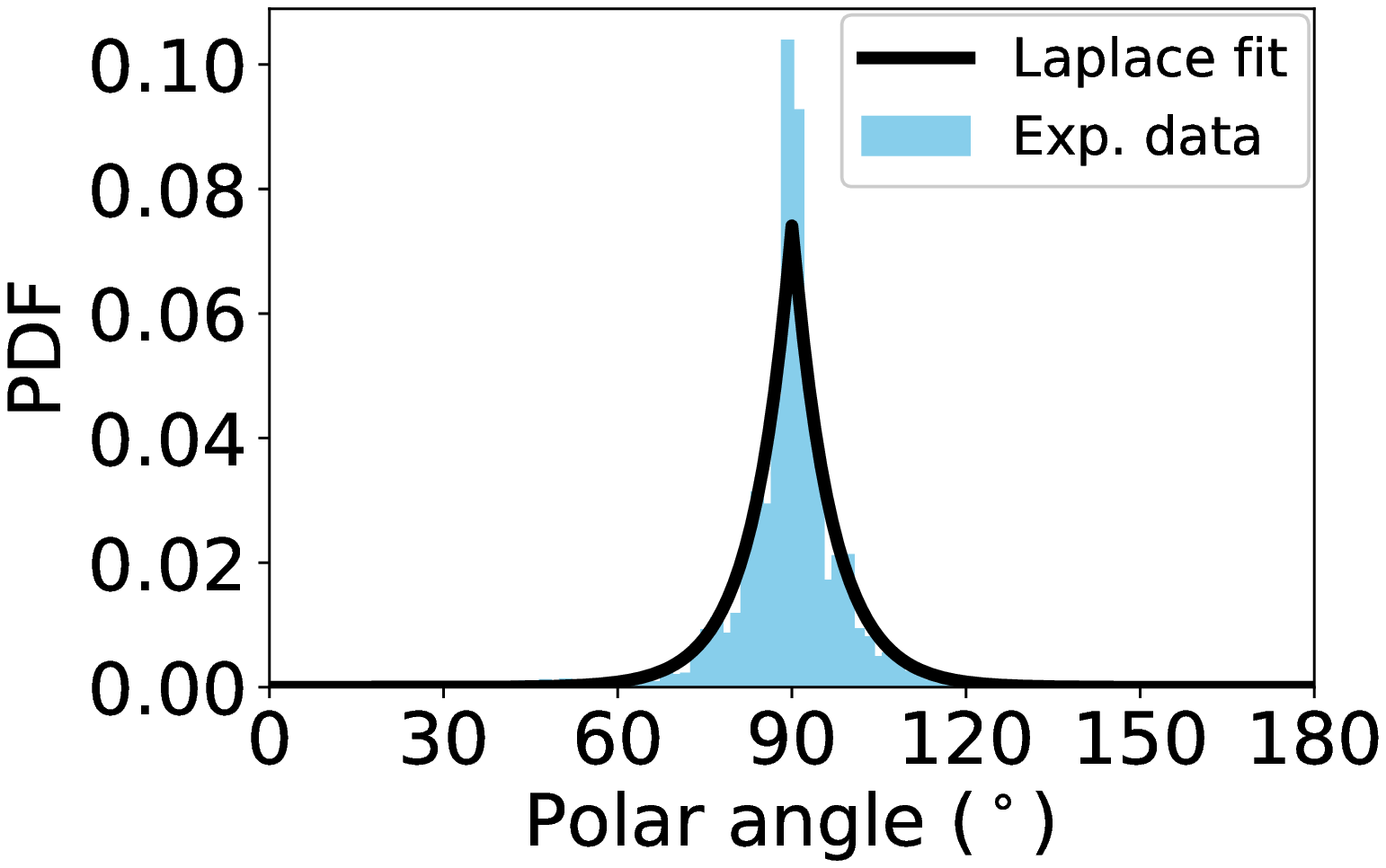}
    \caption{Experimentally obtained 
    $\theta$  (desktop VR, Lite) is Laplace distributed.}
  \label{polar_fit_our_data}
  \end{minipage}
  \hspace{0.1cm}
  \begin{minipage}{0.74\linewidth}
\centering
 \vspace{-0.2cm}
\hspace{-0.2cm}
\begin{subfigure}{.32\textwidth}
  \centering
  \includegraphics[width=1\linewidth]{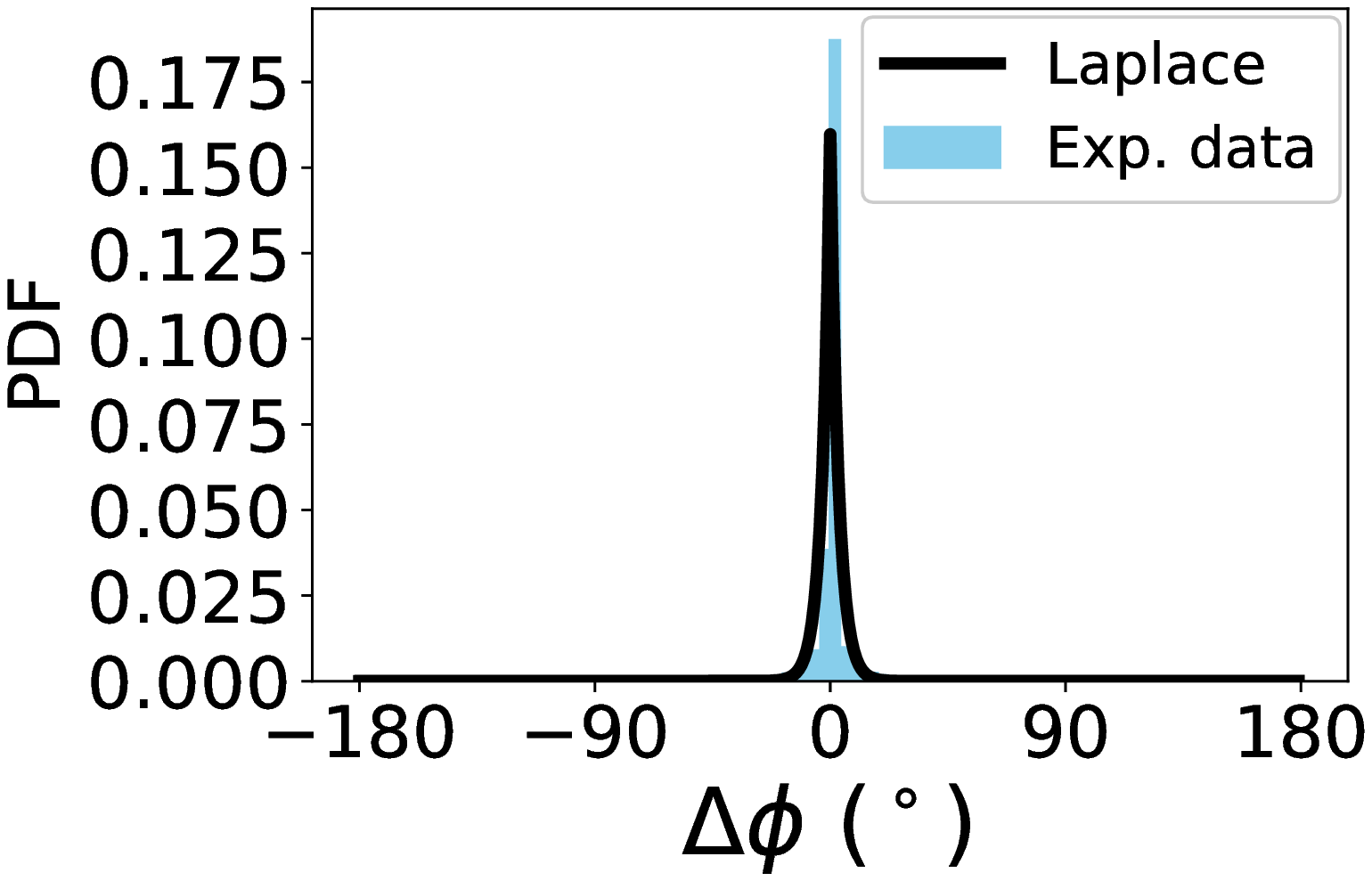}  
   \vspace{-0.5cm}
  \caption{$\Delta t=10/60 $ s}
  \label{fig:sub-first}
\end{subfigure}
   \vspace{-0cm}
  \hspace{-0.2cm}
\begin{subfigure}{.32\textwidth}
  \centering
  \includegraphics[width=1\linewidth]{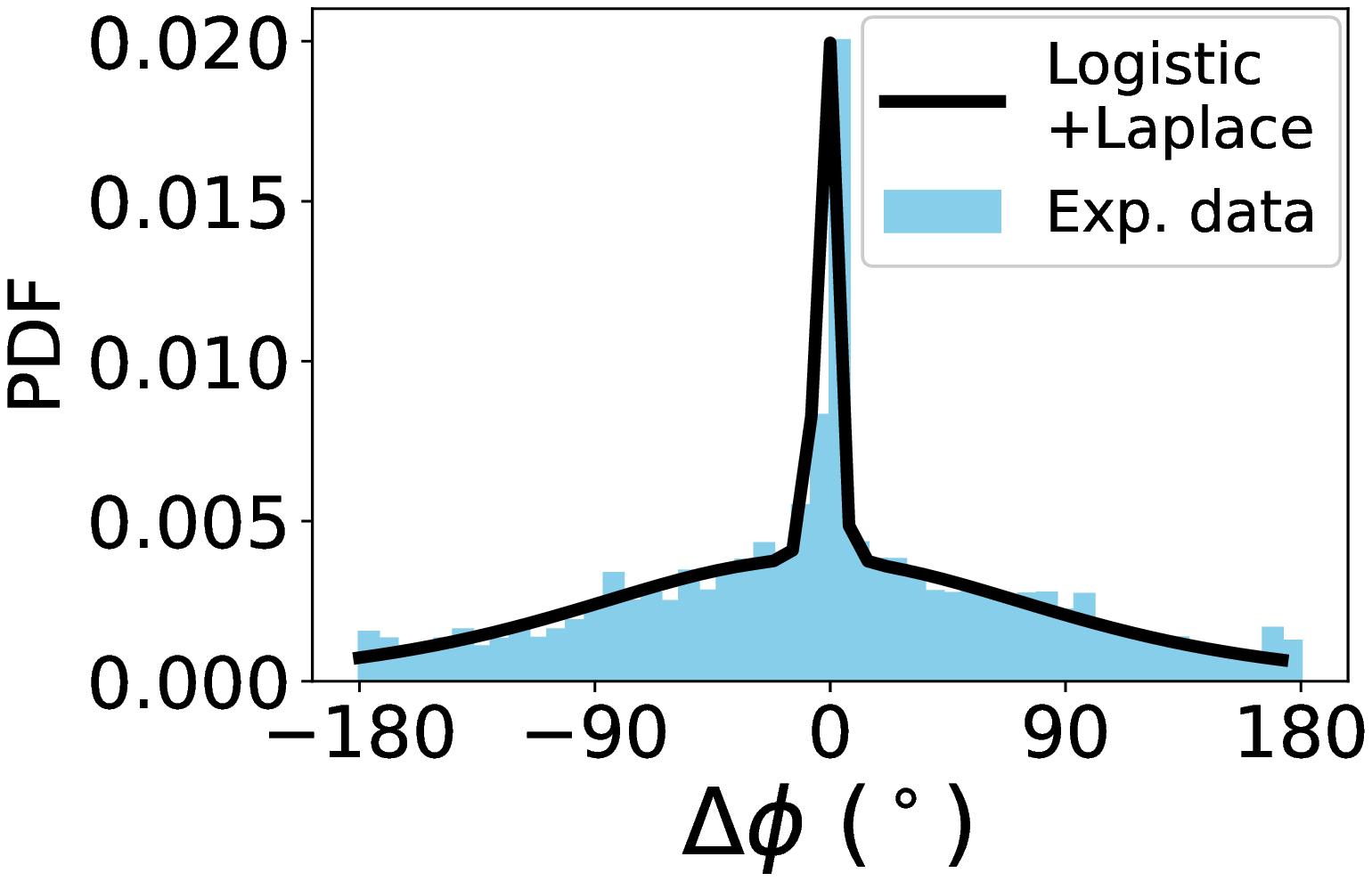}  
   \vspace{-0.5cm}
  \caption{$\Delta t=500/60 $ s}
  \label{fig:sub-first}
\end{subfigure}
   \vspace{-0cm}
  \hspace{-0.2cm}
 \begin{subfigure}{.32\textwidth}
  \centering
  \includegraphics[width=1\linewidth]{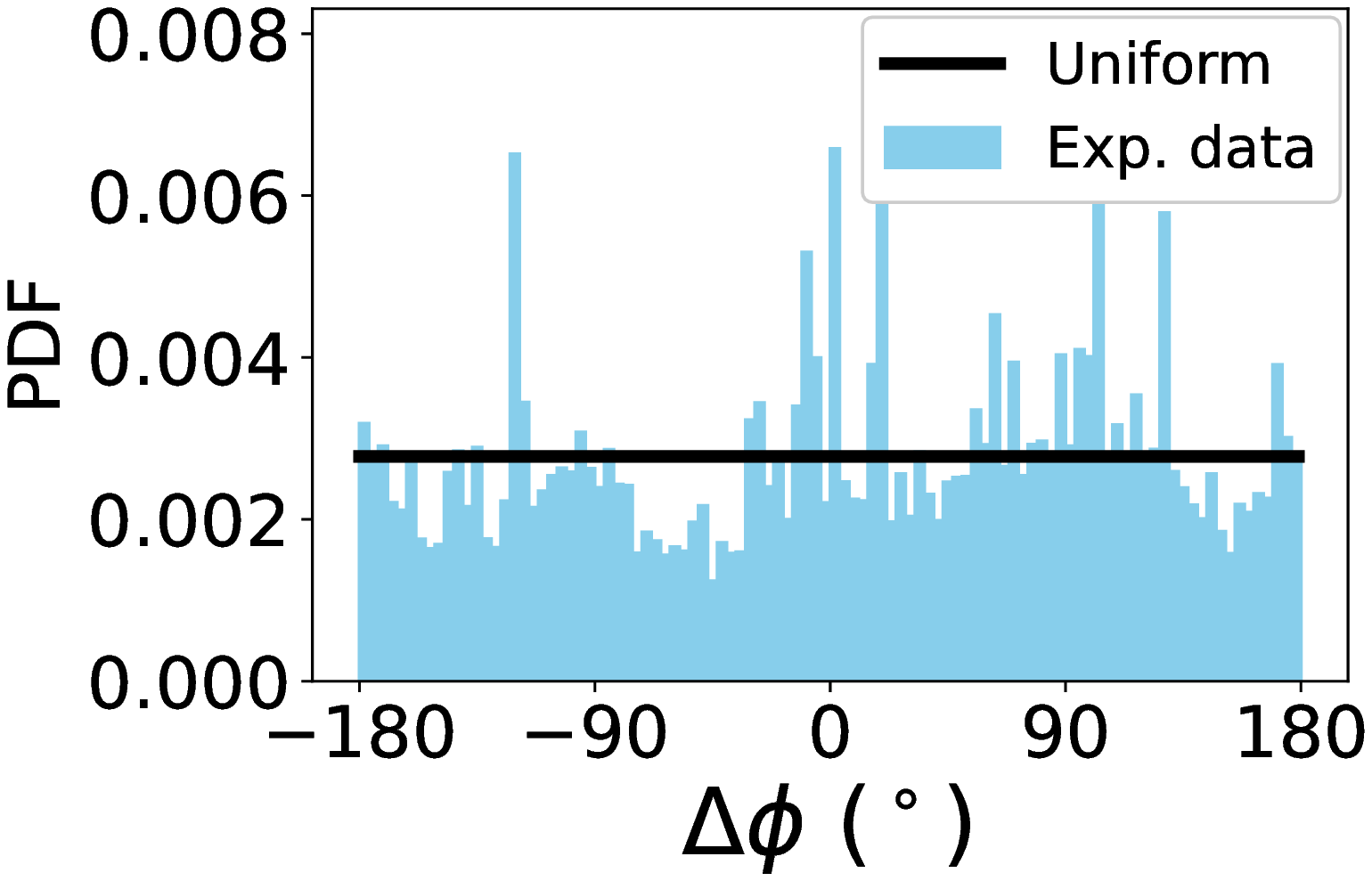}  
   \vspace{-0.5cm}
  \caption{$\Delta t=2000/60 $ s}
  \end{subfigure}
\caption{The distribution of the experimental data for azimuth change $\Delta \phi$ after $\Delta t$, and the best distribution fit among common distributions and mixed distributions.
} 
\label{fig:azimuthdifference_our}
    \end{minipage}
   \vspace{-0.2cm}
\end{figure*}

We evaluate the distribution fit for the experimental measurements of $\theta$, $\Delta \theta$, and $\Delta\phi$.  
In analyzing ${p_\theta }(\theta )$ and ${p_{\Delta \theta} }({\Delta \theta})$, we fit the experimental data to a set of common statistical distributions.
In analyzing ${p_{\Delta\phi} }(\Delta\phi )$, the PDFs of which have irregular shapes for some values of $\Delta t$, we fit the experimental data to the set of common distributions and mixed distributions of two different common distributions. 
As the error metric, we use the sum of squared errors (SSE)~\cite{statisticsDictionary} between the data and the fitted distribution.

\subsubsection{Statistical distribution of the polar angle $\theta$}  
We found \emph{Laplace distributions}, with means close to 90$^\circ$ (89.0$^\circ$--92.3$^\circ$) and scales ranging from 3.5 to 7.4, to best fit experimental  
data (see a summary 
in Table~\ref{our_polar_model} and a fit example in Fig.~\ref{polar_fit_our_data}). 
This is intuitive: 
it corresponds 
to humans having a 
bias for looking straight ahead, to the central parts of VR contents, without frequently tilting their heads.
 Among the 3~VR~games, the scale values are the largest for VK (4.3--7.4), which has more contents scattered along different latitudes than the other games. 
 Among the 3 interface types, the scale values are the smallest for headset VR (3.5--4.3), 
corresponding to users looking straight ahead rather than up and down. We hypothesize that this is due to 
the discomfort associated with 
tilting the head drastically while wearing a headset.

\subsubsection{Statistical models for polar angle change $\Delta \theta$} 
The experimental distributions of $ {\Delta \theta }$ for different $\Delta t$ values closely fit zero-mean Laplace distributions. 
The scales $b_{1,\theta}$ of the Laplace distributions that yield the best fit for different $\Delta t$ values in Lite are shown in Table~\ref{Table Polar Angle Change}. 
As  $\Delta t$ increases, the correlation between polar angles decreases, leading to the increase of the scale with $\Delta t$ (e.g., from 0.305 for $\Delta t = 5/60$~s to 3.407 for $\Delta t = 100/60$~s for desktop VR). 
Among the 3 interface types, headset VR has the largest $b_{l,\theta}$ when $\Delta t$ is small (e.g., 
1.92 vs.~1.48 and 1.41 for $\Delta t = 30/60$~s), 
indicating that 
the polar angle changes more rapidly. 
This is due to the ease of changing viewport orientation over a small time interval in headset VR.

\subsubsection{Statistical models for azimuth angle change $\Delta \phi$} In our examinations, for some $\Delta t$ the distributions of $\Delta \phi$ appeared to have canonical shapes, while for others they appeared as a mixture of distributions. Thus we fit the experimental data to both common and mixed distributions. We present the best distribution fits and their parameters, for a subset of $\Delta t$ values, in desktop VR for all 3~games jointly, in Table~\ref{tab:model_of_azimuth_our}. 
For the cases of mixed distributions, logistic and Laplace in these examples, the PDF of the mixed distribution 
is written as
\begin{equation*}
\label{eqn:mixture}
\begin{array}{l}
{f_{{\rm{mixed}}}}\left( x \right) = {p_l}\frac{1}{{2{b_l}}}\exp \left( { - \frac{{\left| x \right| - {\mu _l}}}{{{b_l}}}} \right)\frac{1}{{1 - \exp \left( { - \frac{{180 - {\mu _l}}}{{{b_l}}}} \right)}} +\\
 \left( {1 - {p_l}} \right)\frac{{\exp \left( { - \frac{{\left| x \right| - {\mu _{lo}}}}{{{b_{lo}}}}} \right)}}{{{b_{lo}}{{\left( {1 + \exp \left( { - \frac{{\left| x \right| - {\mu _{lo}}}}{{{b_{lo}}}}} \right)} \right)}^2}}}\frac{1}{{\left( {2{{\left( {1 + \exp \left( { - \frac{{180 - {\mu _{lo}}}}{{{b_{lo}}}}} \right)} \right)}^{ - 1}} - 1} \right)}}
\end{array}
\end{equation*}
where $\mu_l$ and $b_l$ are the mean and the scale of the Laplace distribution, $\mu_{lo}$ and $b_{lo}$ are the mean and the scale of the logistic distribution, and $p_l$ is used to alter the fractions of the logistic  and the Laplace distributions. 

{\it The best distribution fit for $\Delta\phi$ changes with $\Delta t$.} When $\Delta t$ is small (i.e., when $\Delta t < \beta_1$), $\Delta\phi$ is best modeled by a Laplace distribution with a relatively small scale.  When  $\beta_1 \leqslant \Delta t < \beta_2$ , $\Delta\phi$ is best modeled by a mixture of logistic and Laplace distributions, corresponding to users' tendency 
to change their head orientations only slightly over these time intervals (i.e., $-15^\circ<\Delta \phi<15^\circ$). Finally,  when   $\Delta t \geqslant \beta_2$, the individual angle observations become uncorrelated and are best modeled by a uniform distribution $\mathcal{U}_{\left[ { - {{180}^\circ },{{180}^\circ }} \right)}$. From the collected desktop VR pose data, we obtain $\beta_1=189/60$~s and $\beta_2=1549/60$~s. Examples of these three cases 
are shown in Fig.~\ref{fig:azimuthdifference_our}. 
For the other 2 VR interfaces, we observe similar patterns, but $\beta_1$ and $\beta_2$ are different: $\beta_1=244/60$~s and $\beta_2=1003/60$~s for headset VR,  $\beta_1=496/60$~s and $\beta_2=1006/60$~s for phone-based VR.

\begin{table}[t]
\caption{
Distributions that best fit the experimental 
azimuth angle change $\Delta \phi$ over $\Delta t$, and the SSE between the experimental and the fitted distributions.} 
\vspace{-0.1cm}
\label{tab:model_of_azimuth_our}
\begin{tabular}{c|c c c}
\hline
$\Delta t $ (s)&  \multicolumn{3}{|c}{Best fit} \\
 & Distribution & Parameters& SSE \\
\hline
10/60  & Laplace&$\mu_l=0.0$, $b_l=3.130$ &$1.465 \times {10^{ - 3}}$\\
60/60  & Laplace& $\mu_l=0.0$, $b_l=13.868$&$7.016\times {10^{ - 3}}$ \\
200/60&   \multirow{2}{3em}{Logistic+\\Laplace}&  \multirow{3}{11em} {$\mu_{lo}=-0.1$, $b_{lo}=28.54$, $\mu_{l}=0.1$, $b_{l}=0.24$, \\$p_l=0.36$}& $3.650\times {10^{ - 3}}$ \\
\vspace{-0.35cm}
\\ \\ \\
500/60 &  \multirow{2}{4em}{Logistic+\\Laplace} &  
\multirow{3}{11em}{$\mu_{lo}=-0.4$, $b_{lo}=53.35$,  \\$\mu_{l}=4.2, b_{l}= 0.34$, \\$p_l=0.13$  }&$2.499\times {10^{ - 3}}$\\
\vspace{-0.35cm}
\\ \\ \\
2000/60 &  Uniform & $\Delta \phi  \sim \mathcal{U}_{\left[ { - {{180}^ \circ },{{180}^ \circ }} \right)}$& $1.739\times{10^{ - 4}}$     \\
\hline
\end{tabular}
\vspace{-0.7cm}
\end{table}

%% file: positionModel.tex
\subsubsection{MRWP model with random pause times}
\label{rwp_exp}

\begin{figure*}[h]
\centering
\begin{minipage}[t]{0.235\linewidth}
   \vspace{0pt}
\centering
   \includegraphics[width=\textwidth]{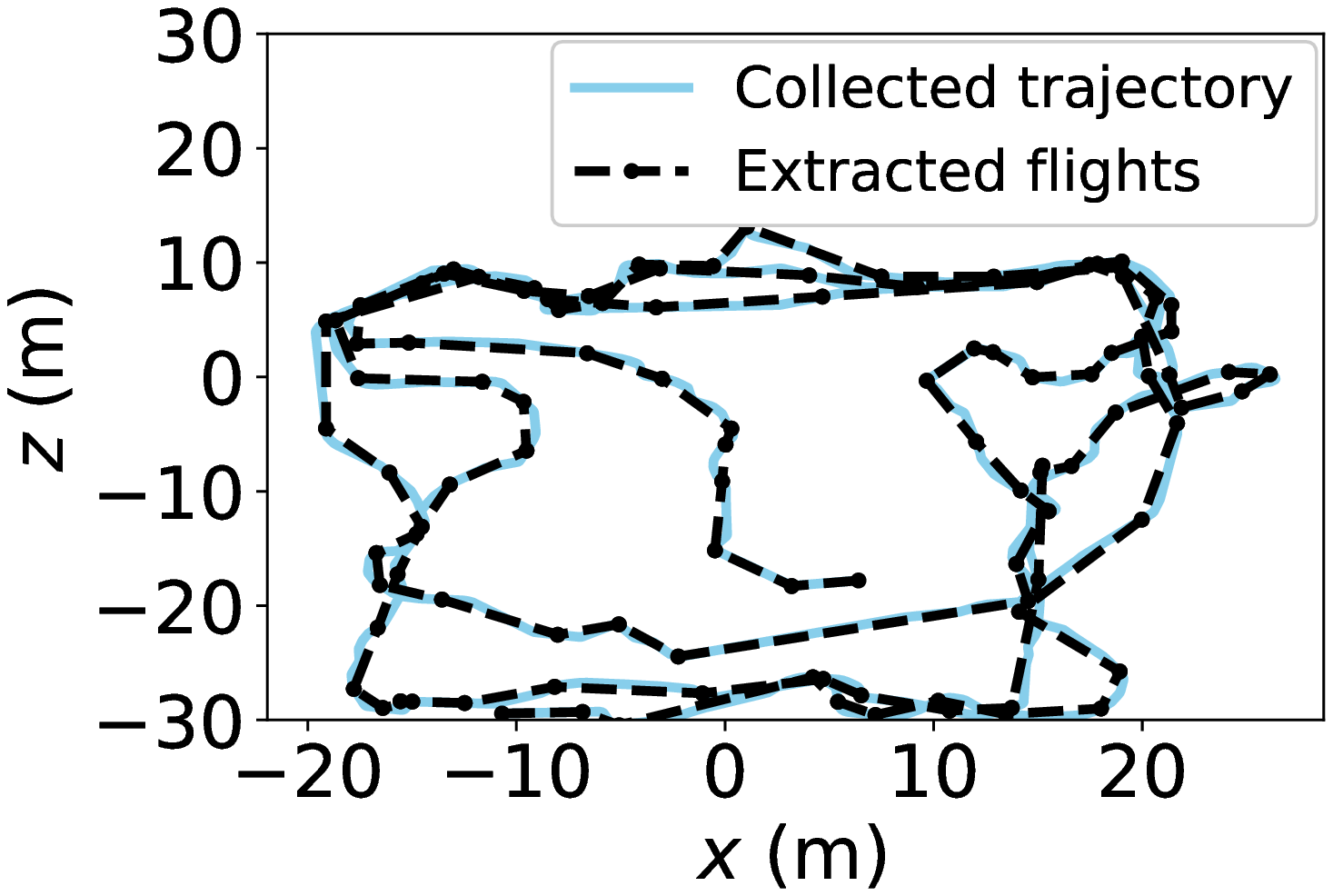}
       \vspace{-0.7cm}
  \caption{The collected trajectory for one user in Lite and the extracted flights.  
  }
   \label{flight_sample}
\end{minipage}
\hspace{-0.02cm}
\begin{minipage}[t]{0.235\linewidth}
   \vspace{0pt}
\centering
\includegraphics[width=\textwidth]{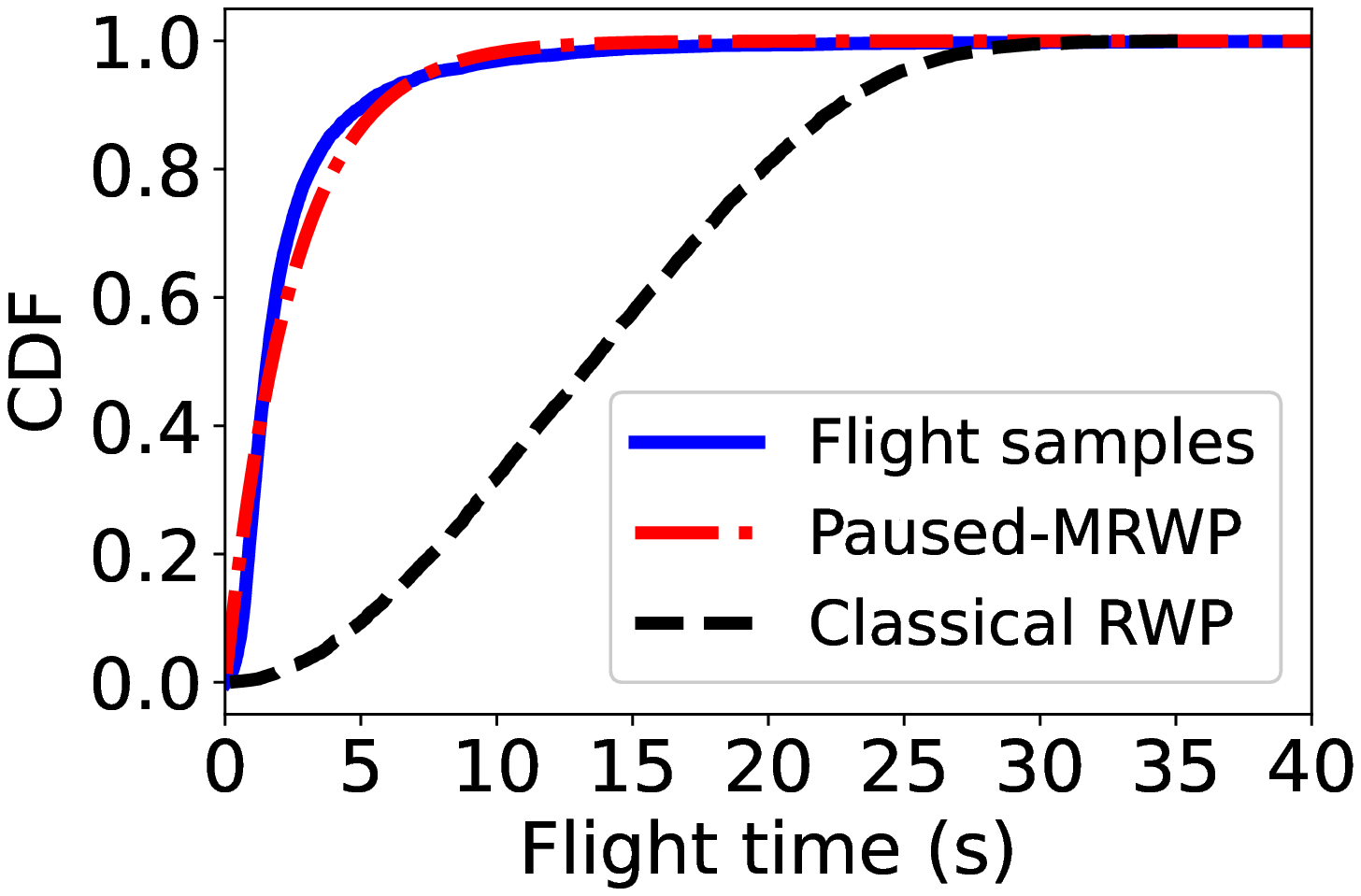}
    \vspace{-0.7cm}
\caption{CDF of the flight time for collected  samples, paused-MRWP, and classical RWP.}
\label{VKflight}
\end{minipage}
\hspace{0.3cm}
 \begin{minipage}[t]{0.235\linewidth}
    \vspace{0pt}
    \centering
\includegraphics[width=\textwidth]{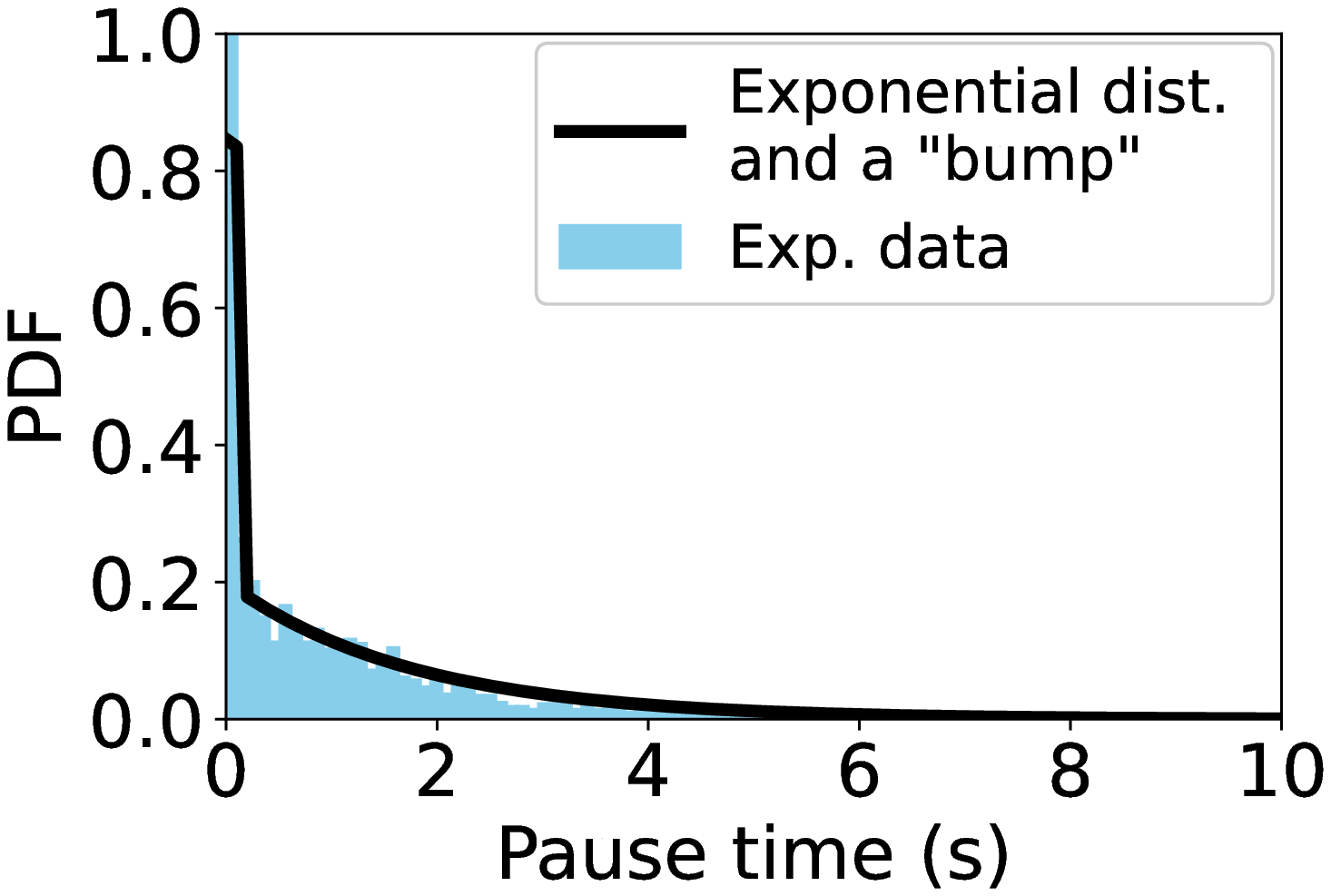}
    \vspace{-0.7cm}
\caption{The pause time for all users and all games. 
}
\label{pausetime_distribution}
\end{minipage}
\hspace{-0.02cm}
\begin{minipage}[t]{0.24\linewidth}
   \vspace{0pt}
\centering
  \includegraphics[width=\textwidth]{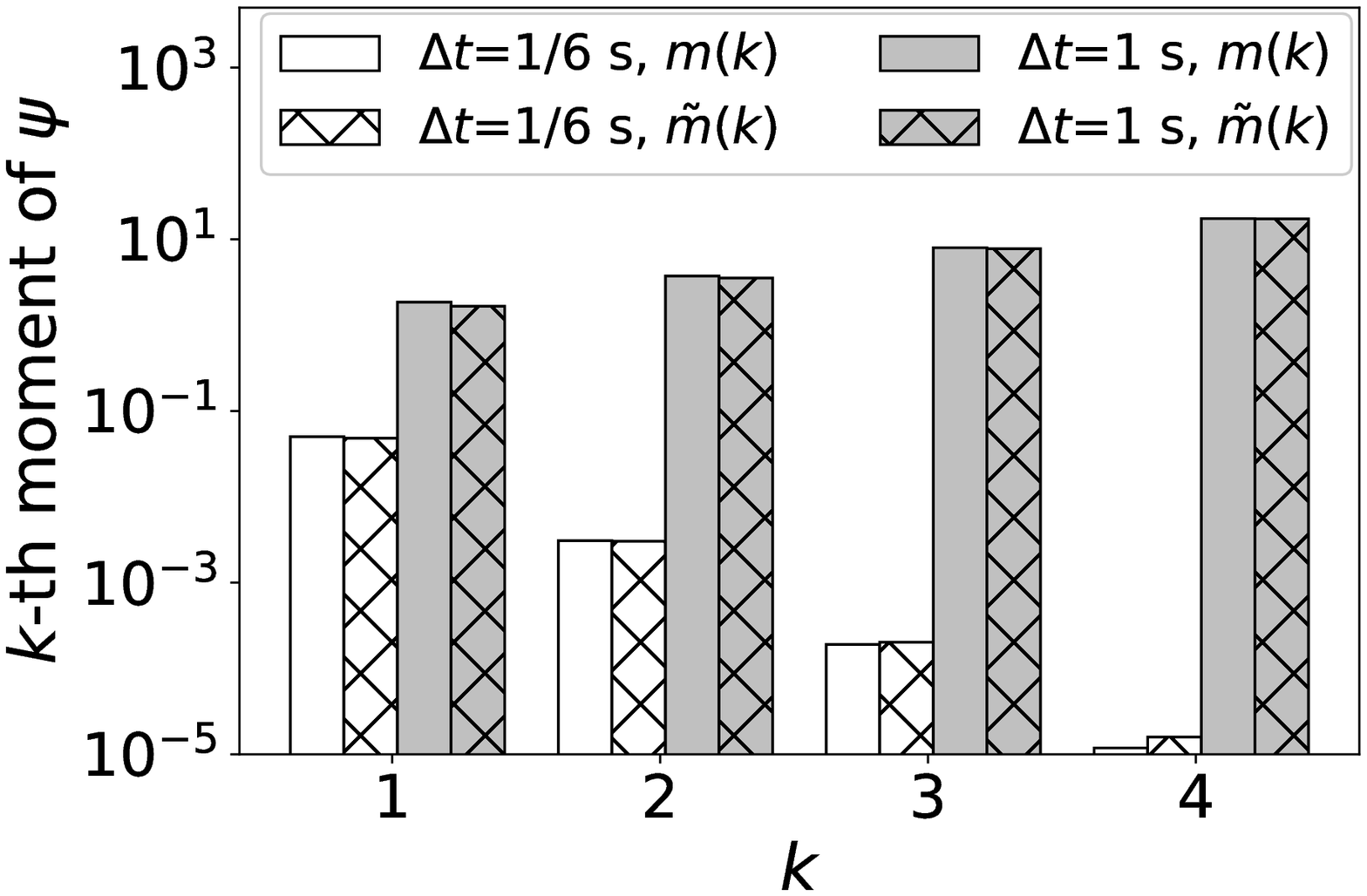}
    \vspace{-0.7cm}
    \caption{The analytical and empirical results of the $k$-th moment of position displacement.
    }
  \label{fig:moment}
  \end{minipage}
 \vspace{-0.4cm}
\end{figure*}

In this section we introduce our 
model for VR viewport position. 
We focus on viewport position change over 
time interval $\Delta t$, in order to analyze the ViS of VR frames that are $\Delta t$ apart.

Adopting the axis notation 
common in computer graphics{ \cite{gameEngineDevelopment,Unity20}}, the viewport positions in the Earth coordinates $XYZ$ (shown in Fig.~\ref{fig:deviceorientation}) are denoted as $\left(x,y,z\right)$, where $y$ is the height of the viewport, and $x$ and $z$ are the coordinates of the viewport positions in the $XZ$-plane (i.e., the ground plane). We assume that $y$ is constant, e.g., $y$ can be set to the human eye level. While changing $y$ is important in some specific contexts, such as exergames \cite{Supernatural,FITXR}, \emph{$y$ is fixed in the vast majority of typical VR experiences}, and in native Oculus Integration app development, to avoid disorienting the users when they sink below or float above the ground in the virtual environment~\cite{OculusDevelopment}.

To model the change of $x$ and $z$, we propose a \emph{paused-MRWP} position model in the infinite plane, 
based on our collected pose data and the modified RWP~\cite{XLin13}.  The model consists of an infinite sequence of points  $ {{W_n}}, {n \in  \mathbb{N}^{+} }$, called \emph{waypoints}, the pause time $S_n$ at each waypoint, the duration $T_n$ to move along a straight line from $W_{n}$ to $W_{n+1}$ with a constant velocity $v$, and the included angle $\alpha_n$ between  $\overrightarrow{W_{n }{W_{n+1}}}$ and the abscissa. The waypoint is expressed as ${W_n}=\left(x_n,z_n\right)$ where $x_n$ and $z_n$ are the coordinates of waypoints in the $XZ$-plane.  The vector $\overrightarrow{W_{n }{W_{n+1}}}$ is called the $n$-th \emph{flight}, and $\alpha_n$ is called the direction of the $n$-th flight. At time 0, the viewport is at $W_1$, and starts to move towards $W_2$. Due to the constant velocity, the flight time $T_n$ for the $n$-th flight is proportional to its length ${\left\| \overrightarrow{W_{n }{W_{n+1}}} \right\|}$. The reason for assuming a constant velocity is that although acceleration can be used to produce more realistic movements, constant-velocity VR movement is known to be more comfortable than the movement with acceleration or deceleration~\cite{VRVelocity}. 
We further assume that $T_n$, $ S_n$, and $ \alpha_n$ are all i.i.d. distributed over $n$. Based on the 
experimental data in desktop VR, we will propose the models for $T_n$, $S_n$, and $ \alpha_n$.

    The collected data shows that the viewpoint movement in the $XZ$-plane is well approximated by a sequence of flights. We apply the standard angle model proposed in \cite{I.Rhee11} to extract flights from the 
    trajectories.
    Fig.~\ref{flight_sample} plots the 
    trajectory in $XZ$-plane of one user in Lite and the extracted flights. Although the viewport does not move in a perfectly straight line during each 
    flight, the trajectory is close to it. 

\textbf{Modeling flight duration.} Our experimental data 
demonstrates that $T_n$ is exponentially distributed, and confirms that the paused-MRWP  better models the flight time in VR than 
the classical RWP models \cite{C.Bettstetter03}. { Specifically, the PDF of the flight times $T_n$, denoted as $f_{T_n}\left(t\right)$, is modeled as $\mu {e^{ - \mu {t}}}$.} Fig.~\ref{VKflight} shows the cumulative distribution function (CDF) of flight times for our collected flight samples, the paused-MRWP model with the best fitted $\mu$, and the classical RWP with a constant velocity in VK game. We see that the flight times of the paused-MRWP model match the measurements better statistically, with the SSE  as low as 0.0021. 
The exponential distribution with the best fitted $\mu$ will be used to model the flight times 
in \S\ref{sec:visibilitysimilarity} and \S\ref{sub:evaluation}. 
 
\textbf{Modeling pause time.} We model the pause time $S_n$ according to the collected data. Fig.~\ref{pausetime_distribution} shows that the exponential distribution with a ``bump" around zero is a good fit to its distribution. 
The PDF of the pause times, denoted as $f_{S_n}\left(s\right)$, is modeled as $\left( {1 - c} \right)\lambda {e^{ - \lambda {s}}} + c\delta\left(s\right)$, where $\lambda {e^{ - \lambda {s}}}$ stands for the PDF of the exponential distribution with parameter $\lambda$, Dirac delta function $\delta\left(s\right)$ models the ``bump" around zero pause time, and $c$ represents the fraction of the ``bump". 

\textbf{Modeling flight direction.} From our pose dataset, the flight angles $\alpha_n$ follow the uniform distribution on $\left[0, 2\pi \right)$.

 With the developed paused-MRWP model, we will focus on observation intervals of duration $\Delta t$, and derive an analytical expression for the moment generating function (MGF) of the displacement between two viewport positions.

\begin{definition}[MGF of position displacement]
\label{def:position displacement}
 Consider an observation interval $\left[t_s, t_s + \Delta t \right]$, where $t_s$ is a sample from $\mathcal{U}_{\left[ {0,T} \right]}$. Let ${\bf X}_{\rm ref}$ and ${\bf X}_{\rm nov}$ be the viewport positions at $t_s$ and $t_s+\Delta t$, respectively. 
 Then, the position displacement $\psi $ is defined as $ {{{\left\| {{{\bf{X}}_{{\rm{nov}}}} - {{\bf{X}}_{{\rm{ref}}}}} \right\|}^2}} $, and the MGF of $\psi$, denoted by $\mathcal{M}_{\psi}\left(\tau\right)$, is given by $\mathcal{M}_{\psi}\left(\tau\right) = \lim_{T \to \infty} \mathbb{E}\left[\exp \left( \tau \cdot \left\| \mathbf{X}_{\rm nov} - \mathbf{X}_{\rm ref}\right\|^2\right)\right]$. 
\end{definition}

We note that during a small $\Delta t$ 
in VR systems, $\mathbf{X}_{\rm ref}$ and $\mathbf{X}_{\rm nov}$  are not necessarily at the waypoints. This is in contrast to conventional applications of RWP models, 
where 
 the movement is observed at a larger timescale \cite{I.Rhee11,C.Bettstetter03}.

\subsubsection{Analysis}

In analyzing $\mathcal{M}_{\psi}\left(\tau \right)$, there are two mutually exclusive and exhaustive cases to consider. Setting $T_0=S_0=0$ as the auxiliary variables, the \emph{Case 1} is the case that $\exists j \in \mathbb{N}^{+}$, $\sum\limits_{n = 0}^j {{T_n} + \sum\limits_{n = 0}^{j - 1} {{S_n}} } <t_s< \sum\limits_{n = 0}^j {{T_n} + \sum\limits_{n = 0}^j {{S_n}} } $, i.e., we start observing the process when the movement is paused. Let $\Lambda$ denote the event that Case 1 holds. Let $\Lambda^\prime$ denote the complement of $\Lambda$, and $\Lambda^\prime$ is the event that \emph{Case 2} holds.
The Case 2 is the case that $\exists j \in \mathbb{N}^{+}$, $\sum\limits_{n = 0}^{j-1} {{T_n} + \sum\limits_{n = 0}^{j - 1} {{S_n}} } <t_s< \sum\limits_{n = 0}^j {{T_n} + \sum\limits_{n = 0}^{j-1} {{S_n}} } $, i.e., we start observing the process during a flight. 
We will first obtain the $k$-th ($k\geqslant 1$) moment of 
$\psi$, $m_\Lambda\left(k\right)$ and $m_{\Lambda^\prime}\left(k\right)$, for Cases 1 and  2 in Lemmas~\ref{displacement_case1} and \ref{displacement_case2}. Combined with the probability that Case 1 
holds given  in Lemma~\ref{Lemma: Probability of Starting During Flight},  we will obtain $\mathcal{M}_{\psi}\left(\tau \right)$ in Theorem~\ref{thm:displacement}.

\begin{lemma} \label{displacement_case1}
Assume Case 1 holds. Let $T^\prime_j=0$, and $S^\prime_j$ be the remaining pause duration after $t_s$ in the same pause interval. Let  $T^\prime_i=T_i$ and $S^\prime_i=S_i$ for $i>j$. We define the events $A_n$ and $B_n$  as
\begin{align*}
{A_n} &= \left\{ \begin{array}{l}
\left\{ {{\xi _1} < 0} \right\},\ n = 1\\
\left\{ {{{\tilde \xi }_{n - 1}}\geqslant 0\ \&\ {\xi _n} < 0} \right\},\ n \geqslant 2 
\end{array} \right. \\
B_n &= \left\{ \xi_n \geqslant 0 \ \& \ \tilde{\xi}_n < 0 \right\}, \  n \geqslant 1
\end{align*}
where  ${\xi _n} \buildrel \Delta \over = \Delta t - \sum\limits_{i = j }^{j + n-1} {{S^\prime_i}}  - \sum\limits_{i = j }^{j + n-1} {{T^\prime_i}}  $ and $\tilde{\xi}_n \triangleq \Delta t -  \sum\limits_{i = j}^{j + n-1} {{S^\prime_i}}  - \sum\limits_{i = j }^{j + n } {{T^\prime_i}}  $. $A_n$ is the event that we end the observation in a pause interval ($\exists {j^\prime} \in \mathbb{N}^{+}$, $\sum\limits_{n = 0}^{j^\prime} {{T_n} + \sum\limits_{n = 0}^{{j^\prime} - 1} {{S_n}} } < t_s + \Delta t < \sum\limits_{n = 0}^{j^\prime} {{T_n} + \sum\limits_{n = 0}^{j^\prime} {{S_n}}} $)
and that there are $n-1$ complete flights in $[t_s, t_s + \Delta t]$. $B_n$ is the event that $\exists j \in \mathbb{N}^{+}$, $\sum\limits_{n = 0}^{j-1} {{T_n} + \sum\limits_{n = 0}^{j - 1} {{S_n}} } <t_s+\Delta t< \sum\limits_{n = 0}^j {{T_n} + \sum\limits_{n = 0}^{j-1} {{S_n}} } $ and that there are $n-1$ complete flights in $[t_s, t_s + \Delta t]$. Then, ${m_\Lambda }\left( k \right)$  is given by 
\begin{eqnarray} \label{Eqn: displacement1}
{m_\Lambda }\left( k \right) = \sum\limits_{n = 2}^\infty  {\mathbb{E}\left[ {{\psi^k}\mathds{1}\left( {{A_n}} \right)} \right]}  + \sum\limits_{n = 1}^\infty  {\mathbb{E}\left[ {{\psi^k}\mathds{1}\left( {{B_n}} \right)} \right]} 
 \end{eqnarray}
 where  
 \begin{equation*}
 \begin{aligned}
{\mathbb{E}\left[ {{\psi^k}\mathds{1}( {{A_n}} )} \right]}&={n+k-2 \choose n-2} \sum\limits_{h = 0}^{n - 2}  {{\mu }{\lambda }{g_{n - 2,h,h + 2,k,2k + n + h + 1}}}, \\
{\mathbb{E}\left[ {{\psi^k}\mathds{1}( {{B_n}} )} \right]}&={n+k-1 \choose n-1} \sum\limits_{h = 0}^{n-1}  {{\lambda }{g_{n - 1,h,h + 1,k,2k + n + h + 1}}} 
 \end{aligned}
 \end{equation*}
and
 ${g_{i,h,n,k,m}} \triangleq \frac{{\left( {2k} \right)!{\mu ^i}{\lambda ^h}{{\left( {\Delta t} \right)}^{m - 1}}{e^{ - \mu \Delta t}}_1{F_1}\left( {n;m; - \left( {\lambda  - \mu } \right)\Delta t} \right)}}{{\left( {m - 1} \right)!}} \times \\ {v^{2k}}{i \choose h}{c^{i-h}}{{\left( {1 - c} \right)}^{h}}$, with ${}_1{F_1}\left(n;m;z\right) $ being the confluent hypergeometric function of the first kind. 
\end{lemma}
\begin{proof}[  Proof]
 See Appendix~\ref{proofdisplacement_case1}.   Proof sketch: Since the collections of events $\left\{A_n\right\}$ and $\left\{B_n\right\}$ are mutually disjoint and collectively exhaustive, we write   ${m_\Lambda }\left( k \right)$ as in \eqref{Eqn: displacement1}. To calculate ${\mathbb{E}\left[ {{\psi^k}\mathds{1}\left( {{A_n}} \right)} \right]}$, we express the \mbox{$k$-th} moment of $\psi$ as ${v^{2k}}{\left\| {\sum\limits_{i = j+1}^{j+n - 1} {{T_i}{{\bf e}_i}} } \right\|^{2k}}$, where ${{\bf e}_i}$ denotes the unit vector whose direction represents the moving direction $\alpha_i$ of the $i$-th flight. Based on the property that 
$\alpha_i$ is i.i.d. uniformly distributed on $\left[0,2\pi\right)$ and on the distributions of flight and pause times, we get the expressions of ${\mathbb{E}\left[ {{\psi^k}\mathds{1}\left( {{A_n}} \right)} \right]}$. Similar techniques are used to obtain  ${\mathbb{E}\left[ {{\psi^k}\mathds{1}\left( {{B_n}} \right)} \right]}$.
\end{proof}

\begin{lemma} \label{displacement_case2}
Assume Case 2 holds. Let $S^\prime_{j-1}=0$, and $T^\prime_j$ be the remaining flight duration after $t_s$ in the same flight interval. Let $T^\prime_i=T_i$ for $i>j$ and $S^\prime_i=S_i$ for $i\geqslant j$. Let the events $A_n$ and $B_n$ be defined as  
\begin{align*}
{A_n} &= \left\{ \begin{array}{l}
\left\{ {{\xi _1} < 0} \right\},\ n = 0\\
\left\{ {{{\tilde \xi }_n}\geqslant0\ \& \ {\xi _{n + 1}} < 0} \right\},\ n\geqslant 1
\end{array} \right. \\
B_n &= \left\{ \xi_n \geqslant 0 \ \& \ \tilde{\xi}_n < 0 \right\}\ n \geqslant 1
\end{align*}
where ${\xi _n} \triangleq \Delta t - \sum\limits_{i = j-1}^{j+n - 2} {{S^\prime_i}}  - \sum\limits_{i = j}^{j+n-1} {{T^\prime_i}} $ and $\tilde{\xi}_n \triangleq \Delta t --\sum\limits_{i = j-1}^{j+n -1} {{S^\prime_i}}  - \sum\limits_{i = j}^{j+n-1} {{T^\prime_i}} $. Then, ${m_{\Lambda^{\prime}}}\left( k \right) $ is given by
\begin{eqnarray*} \label{Eqn: displacement2}
{m_{\Lambda^{\prime}}}\left( k \right) = \sum\limits_{n = 0}^\infty  {\mathbb{E}\left[ {{\psi^k}\mathds{1}\left( {{A_n}} \right)} \right]}  + \sum\limits_{n = 1}^\infty  {\mathbb{E}\left[ {{\psi^k}\mathds{1}\left( {{B_n}} \right)} \right]}
 \end{eqnarray*}
 where 
 \begin{align*}
{\mathbb{E}\left[ {{\psi^k}\mathds{1}\left( {{A_n}} \right)} \right]}= &\frac{{c^n\left( {2k} \right)!{\mu ^n}{{\left( {\Delta t} \right)}^{2k + n}}\exp \left( { - \mu \Delta t} \right)}}{{\left( {2k + n} \right)!}} + \\&\mathds{1}\left(n>0\right) {n + k  \choose n} \sum\limits_{h = 1}^{n } {g_{n,h,h,k,2k + n + h + 1}}, \\
{\mathbb{E}\left[ {{\psi^k}\mathds{1}\left( {{B_n}} \right)} \right]}=& {n + k - 1 \choose n-1}
\sum\limits_{h = 0}^{n - 1}{{\mu }{g_{n - 1,h,h + 1,k,2k + n + h + 1}}} .
 \end{align*}
\end{lemma}
\begin{proof}
The proof is similar to 
that of Lemma \ref{displacement_case1}. 
\end{proof}

Lemmas~\ref{displacement_case1} and \ref{displacement_case2} yield the expressions for the $k$-th moment of the position displacement for arbitrary $\Delta t$.
In VR systems, we 
analyze the ViS for small $\Delta t$. In this case, the number of flights and pauses in the observation interval is limited. The terms of ${\mathbb{E}\left[ {{\psi^k}\mathds{1}\left( {{A_n}} \right)} \right]}$ and ${\mathbb{E}\left[ {{\psi^k}\mathds{1}\left( {{B_n}} \right)} \right]}$, $n\leqslant N$, dominate  ${m_{\Lambda}}\left( k \right) $ and ${m_{\Lambda^{\prime}}}\left( k \right) $. $N=2$ a good choice for $N$ because 
the sum of the terms  ${\mathbb{E}\left[ {{\psi^k}\mathds{1}\left( {{A_n}} \right)} \right]}$ (or ${\mathbb{E}\left[ {{\psi^k}\mathds{1}\left( {{B_n}} \right)} \right]}$), $n\leqslant 2$, accounts for more than
98\% of ${m_{\Lambda}}\left( k \right) $ (or ${m_{\Lambda^{\prime}}}\left( k \right) $) when $t<1$~s and $k\leqslant 4$. Hence, we can simplify the calculation of ${m_{\Lambda}}\left( k \right) $ and ${m_{\Lambda^{\prime}}}\left( k \right) $ by discarding many terms corresponding to the cases of $n>2$ when $\Delta t$ is small (e.g., $\Delta t<1$~s).

\begin{lemma} \label{Lemma: Probability of Starting During Flight}
Let $t_s$ be a sample from $\mathcal{U}_{\left[ {0,T} \right]}$  and $p_T$ be the probability that $t_s$  falls in a flight interval. Then, $p = \mathop {\lim }\limits_{T \to \infty } {p_T}$ exists and is equal to 
$
p = \frac{{\lambda /\left( {1 - c} \right)}}{{\lambda /\left( {1 - c} \right) + \mu }}
$.
\end{lemma}
\begin{proof}
See Appendix~\ref{Appendix: Probability of Starting Druing a Flight}.
\end{proof}

\begin{theorem}[MGF of $\psi$]
\label{thm:displacement}
$\mathcal{M}_{\psi}\left(\tau\right) $ is given by
$ 
\mathcal{M}_{\psi}\left(\tau \right) =\sum_{k=0}^\infty {{m\left(k\right){\tau^k}} \mathord{\left/
 {\vphantom {{\left[ {\left( {1 - p} \right){m_\Lambda }\left( k \right) + p{m_{\Lambda ^\prime}}\left( k \right)} \right]{\psi^k}} {k!}}} \right.
 \kern-\nulldelimiterspace} {k!}}$
where $m(k)$ is the $k$-th moment of the position displacement $\psi$ and $m\left(k\right)=\left[ {\left( {1 - p} \right){m_\Lambda }\left( k \right) + p{m_{\Lambda ^\prime}}\left( k \right)} \right]$.
\end{theorem}

\begin{proof} 
The proof follows directly from Lemmas \ref{displacement_case1}-\ref{Lemma: Probability of Starting During Flight}.  
\end{proof}

Fig.~\ref{fig:moment} compares the empirical results of the $k$-th moment of~$\psi$, ${\tilde m\left(k\right)}$, with  the analytical results $m\left(k\right)$.  We obtain ${\tilde m\left(k\right)}$ by randomly sampling 5000 pairs of viewport positions that are $\Delta t$ apart and calculating the position displacement.
The gap  $\left|{\tilde m\left(k\right)}-m\left(k\right)\right|$   is smaller than $5\times10^{-6}$ when $\Delta t=1/6$~s and $k=4$, and is smaller than 6\% of ${\tilde m\left(k\right)}$ in other cases. 
Theorem~\ref{thm:displacement} will be used to calculate the ViS in \S\ref{sec:visibilitysimilarity}. 

\subsection{Correlations between Orientations and Positions}
\label{sec:correlation}

VR viewports' position and orientation are correlated. Observing the azimuth angles $\phi$ and the walking directions (i.e.,
the included angle between $\overrightarrow{{\bf X}_{\rm ref}{\bf X}_{\rm nov}}$ and the positive direction of X-axis)  in our collected data, we find that the azimuth angles fixate around the walking direction. Similar observations have been made about human walking patterns in non-virtual worlds~\cite{turano2001direction,hollands2002look}. Supported by the pose data, we assume that the azimuth angle at observation start time $t_s$ is the same as the walking direction.

%% file: VisibilitySimilarity.tex
{ We introduce the model for the average visibility similarity given the viewport-to-content distance $d$ ($\ViS\left(d\right)$) in \S\ref{Visibility Similarity Model}.} Then we apply the developed VR pose model to analyze the $\ViS\left(d\right)$ in \S\ref{VisibilitySimilarityAnalysis}, and propose the  ViS-based VR content splitting algorithm, ALG-ViS, in \S\ref{spliting}.

\subsection{ViS Model}
\label{Visibility Similarity Model}
 The analytical model for $\ViS\left(d\right)$ we develop in this section characterizes the statistical average of 
the inter-frame pixel 
similarity over different pose changes given the viewport-to-content distance $d$. 
We define the $\ViS\left(d\right)$ formally 
after introducing the camera model to present how viewport pose determines the rendered pixels. 

\textbf{Camera model.}
In VR, the virtual environment is constructed as computer-generated 3D 
contents, where the pixels in VR frames are generated by capturing the scenes with the camera. 
The camera is modeled as  a standard pinhole camera  following~\cite{Tsai87}.
A 3D point in the virtual environment is projected through the pinhole to a pixel on the VR frame. We denote the camera's angle of view (AoV) as $w_{fv}$. 
{ In addition, we assume that the far plane $d_{fp}$ of the camera,  i.e., the largest viewport-to-content distance beyond which the contents cannot be rendered in the VR frame, is much larger than the viewport position change.}

{Consider two VR frames generated at $t_s$ and $t_s+\Delta t$, 
called reference and novel frames. 
The cameras that capture these 
frames are called reference and novel cameras, respectively. A pixel in the reference frame is projected back to the 3D point in the virtual environment, and the 3D point is projected to the corresponding pixel in the novel frame. The viewport-to-content distance $d$ is formally defined as the distance between the viewport position of the reference frame and the 3D point.  The  $\ViS\left(d\right)$ represents the average similarity of the pixels (of distance $d$) in the reference frame and their corresponding pixels in the novel frame, where the average is taken over different pose changes of reference and novel frames.}

\begin{definition}[$\ViS\left(d\right)$]
\label{def:vis} {Let $\mathcal{S}_{d}$ denote the number of  the reference frame's pixels that are projected to the 3D points with distance $d$ from the reference camera. $d$ is expressed as $d=\sqrt {{{\left( {x_{\rm{3D}} - x_r} \right)}^2} + {{\left( {z_{\rm{3D}} - z_r} \right)}^2}} $, where ${\bf{X}}_{\rm{3D}}=\left(x_{\rm{3D}},z_{\rm{3D}}\right)$, ${\bf{X}}_{\rm{ref}}=\left(x_r,z_r\right)$ are the positions of the 3D point and the reference camera  in the  $XZ$-plane, respectively.   Let $\mathcal{S}_{ViS,d}$, $\mathcal{S}_{ViS,d}\subseteq \mathcal{S}_{d}$, denote the set of the reference frame's pixels that have the same pixel value as the corresponding pixels in the novel frame. 
 The visibility similarity given $d$ is defined as ${\ViS}\left(d\right)\triangleq \mathbb{E}\left[\frac{{\left| {{\mathcal{S}_{ViS,d}}} \right|}}{{\left| {{\mathcal{S}_d}} \right|}}\right]$, which is the average percentage of pixels that  have the same pixel value as the corresponding pixels in the novel frame among ${{\mathcal{S}_d}} $. The mean of $\frac{{\left| {{\mathcal{S}_{ViS,d}}} \right|}}{{\left| {{\mathcal{S}_d}} \right|}}$ is taken over the possible viewport pose changes from reference to novel cameras modeled in \S\ref{sec:pose}. }
 \end{definition}
 
 \begin{wrapfigure}{R}{0.28\textwidth}
  \centering
     \vspace{-0.3cm}
  \includegraphics[width=0.28\textwidth]{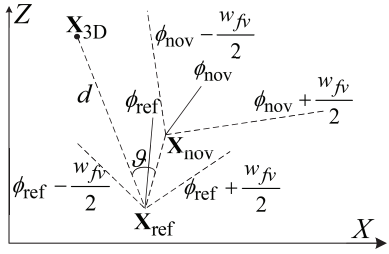}
    \vspace{-0.2cm}
    \caption{ Schematics of camera positions and angles in the $XZ$-plane.}
  \label{fig:fovterm}
  \vspace{-0.2cm}
\end{wrapfigure} 

We break down the $\ViS\left(d\right)$ into two terms: (1) \textit{the FoV term $\ViS_{fov}$}, which  represents the fraction of the VR contents contained in the FoVs of both the novel and the reference cameras, 
 and (2)  \textit{the distance term $\ViS_{dst}\left(d\right)$}, 
 which quantifies 
 the ratio of the number of pixels representing the same 3D points (of distance $d$) in the novel and reference frames. 
 In VR systems, the viewport moving closer to the VR contents will result in the use of more pixels to represent the contents. Note that $\ViS_{fov}$ is independent of $d$ while $\ViS_{dst}\left(d\right)$ is a function of $d$.
 In our analysis, we ignore the influence of occlusion in 
 $\ViS\left(d\right)$, i.e., we do not consider the case where the occluded objects are in the FoV of both reference and novel frames, are rendered in the novel frame, but occluded by the other contents in the reference frame. 
 In our numerical results, we show that their effect is small. { We have $\ViS\left(d\right)=\ViS_{fov}\ViS_{dst}\left(d\right)$. 
 The $\ViS_{fov}$ and $\ViS_{dst}\left(d\right)$ are defined formally below.}

\begin{definition}[FoV term] \label{def:fovterm}
$\ViS_{fov}$ is defined as the statistical average of the multiplication of the  fraction of overlapping polar and azimuth angles of reference and novel cameras. 
\end{definition}

We analyze $\ViS_{fov}$ with the distributions of $\Delta \theta=\theta_{\rm{nov}}-\theta_{\rm{ref}}$ and  $\Delta \phi=\phi_{\rm{nov}}-\phi_{\rm{ref}}$ obtained in \S\ref{sub:orientation}.  Fig.~\ref{fig:fovterm} depicts the positions and angles of reference and novel camera viewports in the $XZ$-plane. Since the far plane $d_{fp}$ is much larger than  $\left\|\overrightarrow{\bf{X}_\text{ref} \bf{X}_\text{nov}}\right\|$, the fraction of the overlapping azimuth angles 
is $\frac{{\min \left( {{\phi _{\rm{ref}}} + \frac{{{w _{fv}}}}{2},{\phi _{\rm{nov}}} + \frac{{{w _{fv}}}}{2}} \right) - \max \left( {{\phi _{\rm{ref}}} - \frac{{{w _{fv}}}}{2},{\phi _{\rm{nov}}} - \frac{{{w _{fv}}}}{2}} \right)}}{{{w _{fv}}}}$. The fraction of the overlapping polar angles is obtained similarly.
Hence, the FoV term is 
\begin{equation}
\label{eq:fov_def}
\begin{aligned}
\ViS_{fov} = \mathbb{E}_{\Delta\theta}\left[1 - \frac{|\Delta\theta|}{w_{fv}}\right] \cdot \mathbb{E}_{\Delta\phi}\left[1 - \frac{|\Delta\phi|}{w_{fv}}\right].
 \end{aligned}
 \end{equation}

\begin{definition}[Distance term] 
\label{distancedefinition}
 $\ViS_{dst}\left(d\right)$ is defined as
\begin{equation}
\label{def:p_d_def}
\ViS_{dst}\left(d\right)\triangleq \mathbb{E}_{\psi,\vartheta }\left[ {{{\left( {\frac{{\sqrt {{d^2} + \psi  - 2d\sqrt \psi  \cos \left( \vartheta  \right)} }}{d}} \right)}^2}} \right]
\end{equation}
where ${\sqrt {{d^2} + \psi  - 2d\sqrt \psi  \cos \left( \vartheta  \right)} }$ is the distance between the VR content and the  novel camera, and $\vartheta$ is the included angle between $\overrightarrow{\bf{X}_{\rm{ref}}{\bf{X}}_{\rm{3D}}}$ and $\overrightarrow{\bf{X}_{\rm{ref}}\bf{X}_{\rm{nov}}}$, as shown in Fig.~\ref{fig:fovterm}. 

 \end{definition}

\subsection{The Analysis of $\ViS\left(d\right)$} \label{VisibilitySimilarityAnalysis}

We first provide a closed-form expression for $\ViS_{ fov}$, and then approximate $\ViS_{dst}\left(d\right)$ tightly with a small error. 

\begin{theorem}[FoV term]
\label{prop:fovterm}
$\ViS_{fov}$ is equal to
\begin{equation*}
\label{eq:fov}
\ViS_{fov}=\left( 1 - \frac{{{b_{l,\theta }} - \exp \left( { - \frac{{{w _{fv}}}}{{{b_{l,\theta }}}}} \right)\left( {{b_{l,\theta }} + {w _{fv}}} \right)}}{{{w _{fv}}}}\right) \left( {1 - p_f^\phi } \right)
\end{equation*}
where ${{b_{l,\theta }}}$ is the scale of the fitted Laplace distribution of $\Delta\theta$, and $p_f^\phi  = \frac{2}{{{w _{fv}}}}\int_0^{{w _{fv}}} {\Delta \phi {p_{\Delta \phi }}\left( {\Delta \phi } \right)d} \left( {\Delta \phi } \right)$, where ${p_{\Delta \phi }}\left(\Delta \phi\right)$ is the PDF of $\Delta\phi$. 
$p_f^{\phi} $ is equal to
\begin{equation}
\label{eq:fov2}
{p_{f}^{ \phi }}=\left\{ \begin{array}{l}
\frac{{{b_l} - \exp \left( { - \frac{{{w _{fv}}}}{{{b_l}}}} \right)\left( {{b_l} + {w _{fv}}} \right)}}{{{w _{fv}}}},\ \Delta t < \beta_1\\
\frac{{2\left( {1 - {p_l}} \right)}}{{{w_{fv}}}}\frac{{{w_{fv}} + {b_{lo}}\log \left( {\frac{2}{{1 + {e^{ - \frac{{{w_{fv}}}}{{{b_{lo}}}}}}}}} \right) - \frac{{{w_{fv}}}}{{1 + {e^{ - \frac{{{w_{fv}}}}{{{b_{lo}}}}}}}}}}{{2{{\left( {1 + {e^{ - \frac{\pi }{{{b_{lo}}}}}}} \right)}^{ - 1}} - 1}} \\
+{p_l}\frac{{{b_l} - \exp \left( { - \frac{{{w _{fv}}}}{{{b_l}}}} \right)\left( {{b_l} + {w _{fv}}} \right)}}{{{w _{fv}}\left( {1 - \exp \left( { - \frac{\pi }{{{b_l}}}} \right)} \right)}},\ \beta_1 \leqslant \Delta t < \beta_2\\
{\frac{{{w_{fv}}}}{{2\pi }}},\ \Delta t > \beta_2.
\end{array} \right.
\end{equation}
\end{theorem}

\begin{proof}[Proof]
See Appendix~\ref{proof_fov}.
\end{proof}

We focus on small $\Delta t$ (e.g., $\Delta t\leqslant 1$~s which belongs to the first case in \eqref{eq:fov2}), as the 
most relevant, in practice, to exploiting VR frame correlation. In this case, the FoV term only depends on  $w_{fv}$,  $b_{l}$, and $b_{l,\theta}$.

\begin{theorem}[Distance term]
\label{distance_term}
For  $\varepsilon>0$, we can approximate $\ViS_{dst}\left(d\right)$ within 
error of $\kappa 
\varepsilon {\mathcal{M}_\psi }\left( { - \frac{1}{{{\varepsilon ^2}}}} \right) 
 $, where $\kappa=\frac{{4\sin \left( {\frac{{{w_{fv}}}}{2}} \right)}}{{{w_{fv}}d\sqrt \pi  }}$. Specifically, $\ViS_{dst}\left(d\right)$ can be approximated according to
\begin{equation}
\begin{aligned}
\label{eq:p_d}
\left| {\ViS_{dst}}\left(d\right)-1-\frac{{m\left( 1 \right)}}{{{d^2}}}+\kappa \sum\limits_{i = 0}^\infty  {g\left( i \right)} \right|
< \kappa \varepsilon 
{\mathcal{M}_\psi }\left( { - \frac{1}{{{\varepsilon ^2}}}} \right)
\end{aligned}
\end{equation}
where $g\left( i \right) = \frac{{{{\left( { - 1} \right)}^i}{\varepsilon ^{ - \left( {2i + 1} \right)}}m\left( {i + 1} \right)}}{{\left( {i + \frac{1}{2}} \right)i!}}$.
\end{theorem}

\begin{proof}
{See Appendix~\ref{p_d_proof}.  Proof sketch:  Taking the average over $\vartheta$ and substituting $m\left(1\right)$,} \eqref{def:p_d_def} is rewritten as $\ViS_{dst}\left(d\right)=1 + \frac{m\left(1\right)}{{{d^2}}} - \frac{{4\sin \left( {\frac{{{w_{fv}}}}{2}} \right)}}{{w_{fv}}d  } {\mathbb{E}_\psi }\left[ \sqrt{\psi}  \right]$. Since ${\mathbb{E}_\psi }\left[ \sqrt{\psi}  \right]=\frac{1}{{\sqrt \pi  }}\mathbb{E}_{\psi}\left[ {\int_0^{\infty} {\psi {e^{ - \tau \psi }}{\tau ^{ - \frac{1}{2}}}} d\tau } \right]$, we then prove that $\mathbb{E}_{\psi}\left[ {\int_0^{\frac{1}{{{\varepsilon ^2}}}} {\psi {e^{ - \tau \psi }}{\tau ^{ - \frac{1}{2}}}} d\tau } \right]=\sum\limits_{i = 0}^\infty  { g\left(i\right)} $ and that the approximation error  $\mathbb{E}_{\psi}\left[ {\int_{\frac{1}{{{\varepsilon ^2}}}}^\infty  {\psi{e^{ - \tau \psi}}{\tau ^{ - \frac{1}{2}}}} d\tau } \right]$ is smaller than $\varepsilon {\mathcal{M}_\psi }\left( { - \frac{1}{{{\varepsilon ^2}}}} \right) $  to conclude the proof. 
\end{proof}

From \eqref{eq:p_d}, $\ViS_{dst}\left(d\right)$ can be approximated by $1   + \frac{m\left(1\right)}{{{d^2}}}- \kappa \sum\limits_{i = 0}^\infty  {g\left( i \right)} $. The approximation error $\kappa 
\varepsilon {\mathcal{M}_\psi }\left( { - \frac{1}{{{\varepsilon ^2}}}} \right) 
 $ can be made arbitrarily small by choosing a small $\varepsilon$.  
{ $\sum\limits_{i = 0}^\infty  {\left| {g\left( i \right)} \right|}$ is convergent since $\left| {g\left( i \right)} \right|$ is bounded by $\frac{{2{\varepsilon^{ - (2i + 1)}}{{(v\Delta t)}^{2\left(i+1\right)}}}}{{i!}}$. Hence, $\sum\limits_{i = 0}^\infty  { {g\left( i \right)} }$ is convergent,  and the sum of the first $30$ terms provides a good approximation with an error 
$\varepsilon {\mathcal{M}_\psi }\left( { - \frac{1}{{{\varepsilon ^2}}}} \right)<0.01$  in \S\ref{sub:evaluation}.}

 The time taken to calculate the $\ViS\left(d\right)$ is dominated by calling the confluent hypergeometric function ${}_1{F_1}\left(n;m;z\right) $ 
in Lemmas~\ref{displacement_case1} and \ref{displacement_case2}. 
The average time, over 100 iterations, for calculating $\ViS\left(d\right)$  on a Lenovo laptop (equipped with an AMD Ryzen 7 4800H CPU and an NVIDIA GTX 1660 Ti GPU) in MATLAB is only 4.4 ms. Hence, the $\ViS\left(d\right)$ can be calculated for each frame in real time.

\setlength{\textfloatsep}{5pt}
 \begin{algorithm}[t]
 \caption{ALG-ViS.}
 \begin{algorithmic}[1]
  \label{alg-vis}
 \STATE $count=0$;
 \FOR{each new frame}
 \IF{$\mod\left(count,R\right)==0$}
 \STATE The frame is selected as a reference frame;
 \STATE Classify all VR contents as foreground contents;
 \ELSE 
 \STATE $d_{tr}\leftarrow d_{fp}$;
  \FOR {($d = 0.5$; $d<d_{fp}$; $d\leftarrow d+0.5$)}
  \STATE Calculate $\ViS\left(d\right)$;
  \IF {$\ViS\left(d\right) \geqslant \ViS_{tr}$}
    \STATE $d_{tr}\leftarrow d$;
    \STATE break;
  \ENDIF
  \ENDFOR
  \STATE Classify the VR contents with $d\leqslant d_{tr}$ as foreground contents, and other contents as background contents;
    \STATE $count\leftarrow count+1$;
  \ENDIF
  \STATE Render the foreground contents. Reuse the background pixels from the  reference frame by view projection;
  \ENDFOR
 \end{algorithmic} 
 \end{algorithm}
 
\subsection{ViS-Based Foreground and Background Content Splitting}
\label{spliting}

{ Based on the average ViS for a given $d$, i.e., $\ViS\left(d\right)$,  we adaptively split the contents to background and foreground, where the background (with a larger $d$) has a high $\ViS\left(d\right)$ and can be reused  to reduce the resource consumption.} To this end, we propose ALG-ViS 
given in Algorithm~\ref{alg-vis}. For every $R$ consecutive frames, the first frame is selected as the reference frame; the remaining $R-1$ frames are the novel frames. In reference frames, all VR contents are classified as foreground contents. In novel frames, we calculate a distance threshold~$d_{tr}$ such that $\ViS\left(d_{tr}\right)=\ViS_{tr}$, where $\ViS_{tr}$ is the  threshold and a $\ViS\left(d\right)$ larger than $\ViS_{tr}$ indicates high similarity of pixels in reference and novel frames. The VR contents that have the distance 
$d\leqslant d_{tr}$ from the reference camera are classified as foreground contents, the other contents -- as background contents. 
VR system renders the foreground contents and reuses the pixels for the background contents from the reference frame by view projection~\cite{YangLi2019,Warp}. 
We only need to  calculate  $\ViS\left(d\right)$ for $R-1$ values of $\Delta t$ when the inter-frame interval is fixed (e.g., when the system supports the full frame rate as in \S\ref{sub:evaluation}), which makes ALG-ViS even more lightweight.

%% file: Evaluation.tex
\begin{figure*}[hbt!]
\begin{minipage}[t]{0.47\linewidth}
\vspace{-3.2cm}
\begin{subfigure}{.495\textwidth}
  \centering
   \includegraphics[width=\linewidth]{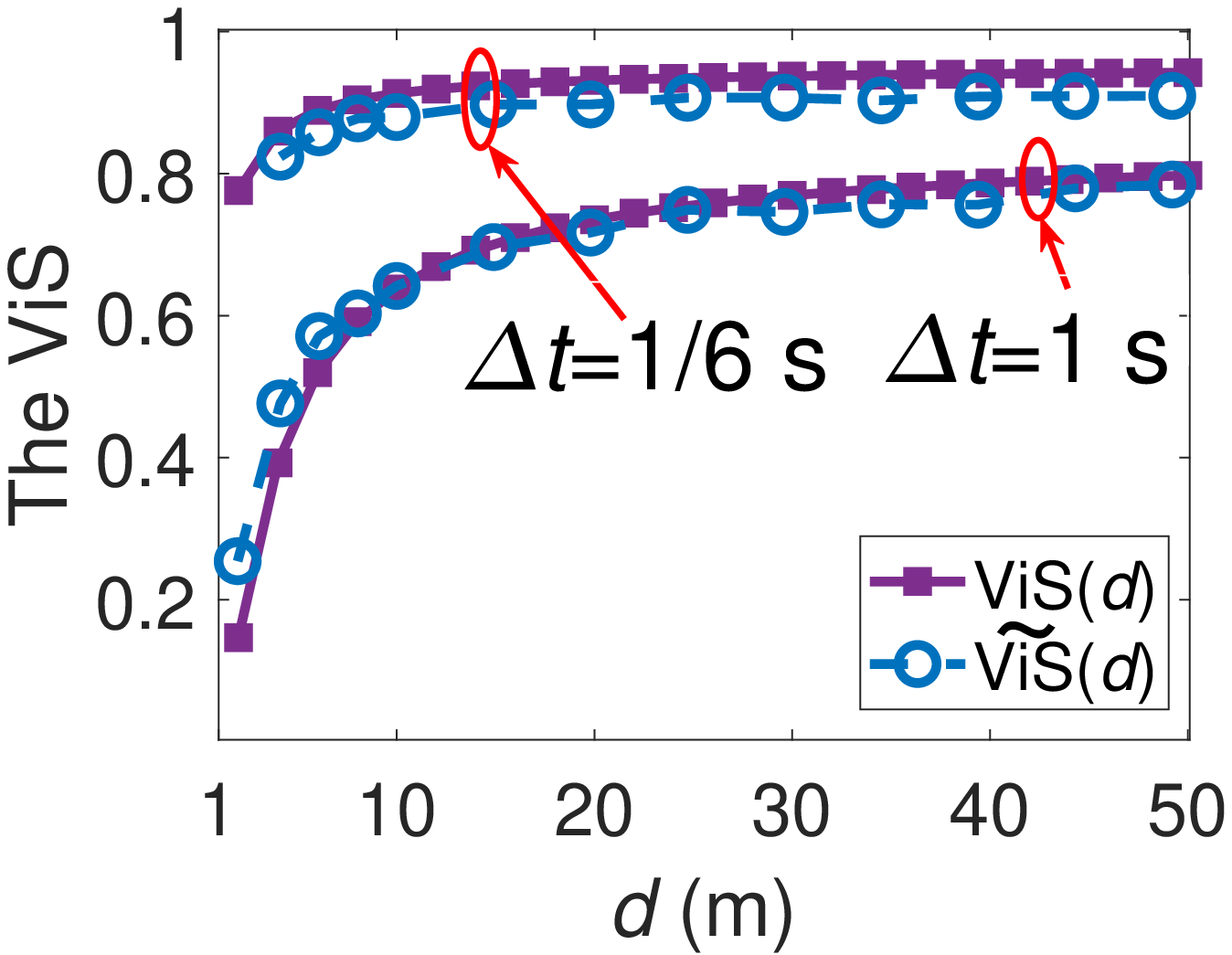}   
  \caption{Headset VR}
  \label{fig:ev_Oculus}
\end{subfigure}
\begin{subfigure}{.495\textwidth}
  \centering
   \includegraphics[width=\linewidth]{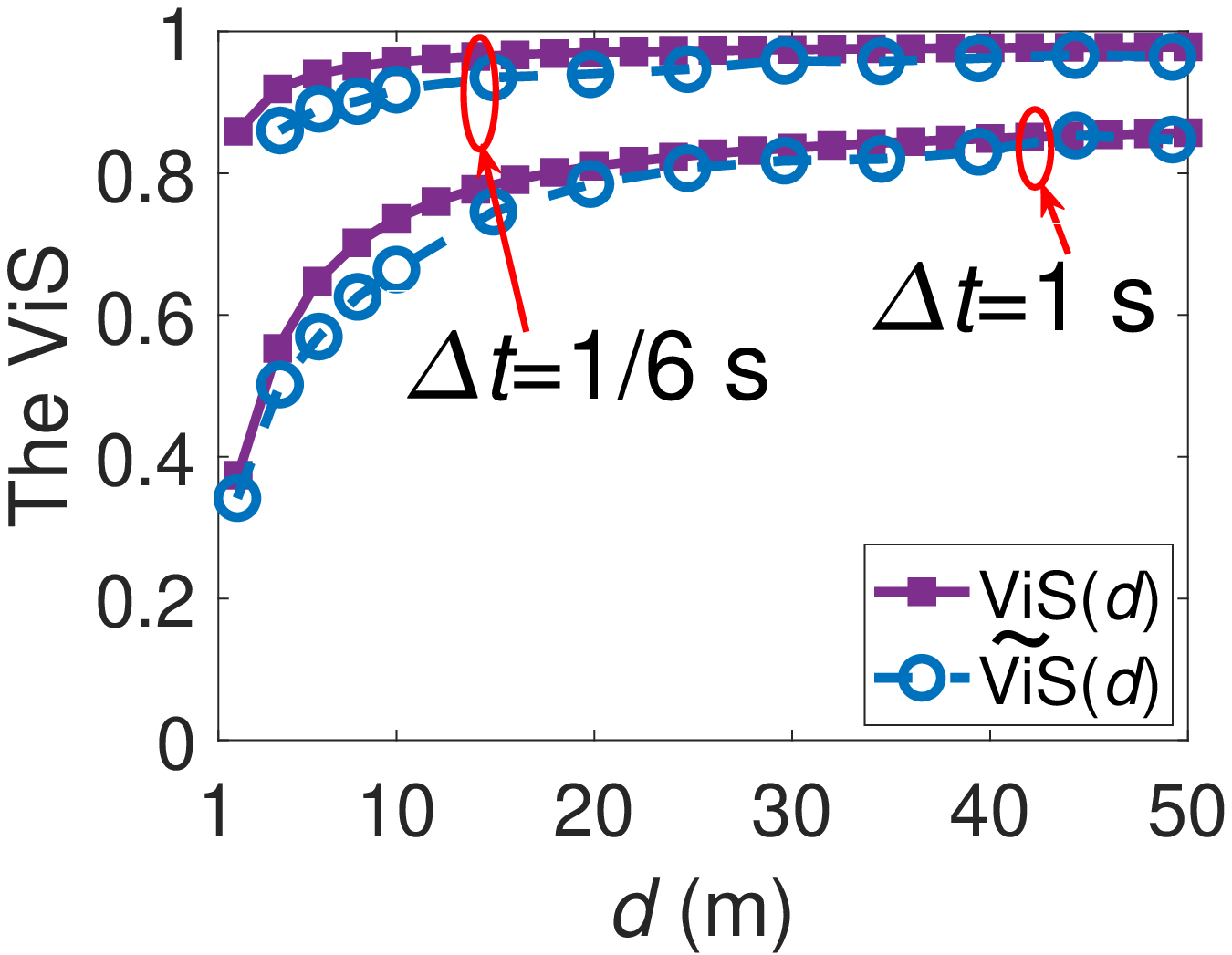}   
  \caption{Phone-based VR}
  \label{fig:ev_phone}
\end{subfigure}
\vspace{-0.15cm}
  \caption{Analytical  and simulation results of the ViS for VK game 
  in
  headset VR and phone-based VR. 
  }
     \label{ev_interface}
\end{minipage}
\hspace{0.1cm}
\begin{minipage}[t]{0.24\linewidth}
\centering
   \includegraphics[width=\textwidth]{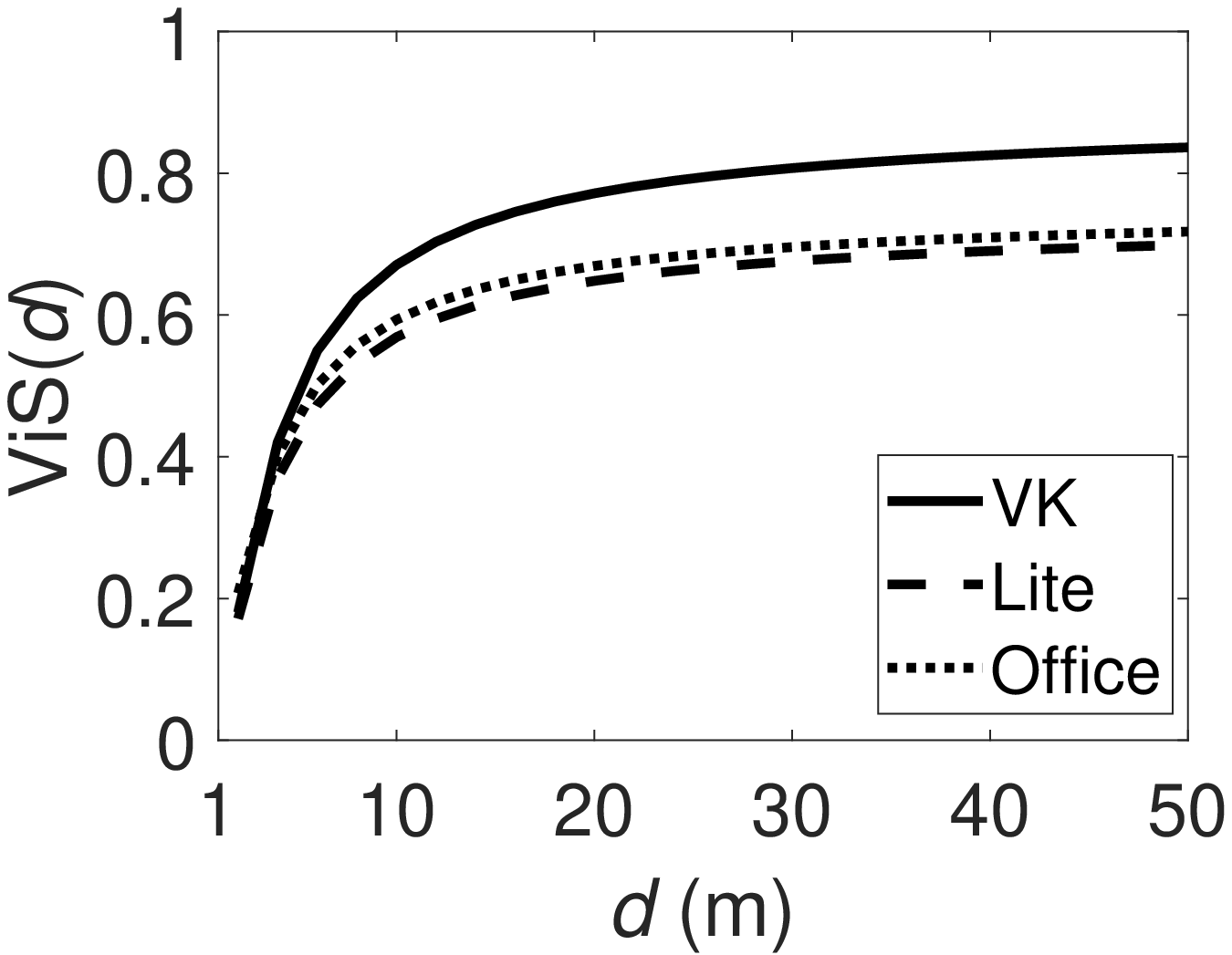}
  \caption{The analytical ViS versus $d$ for different games in desktop VR.}
  \label{ev_game}
\end{minipage}
\hspace{0.1cm}
 \begin{minipage}[t]{0.26\linewidth}
\includegraphics[width=0.9\textwidth]{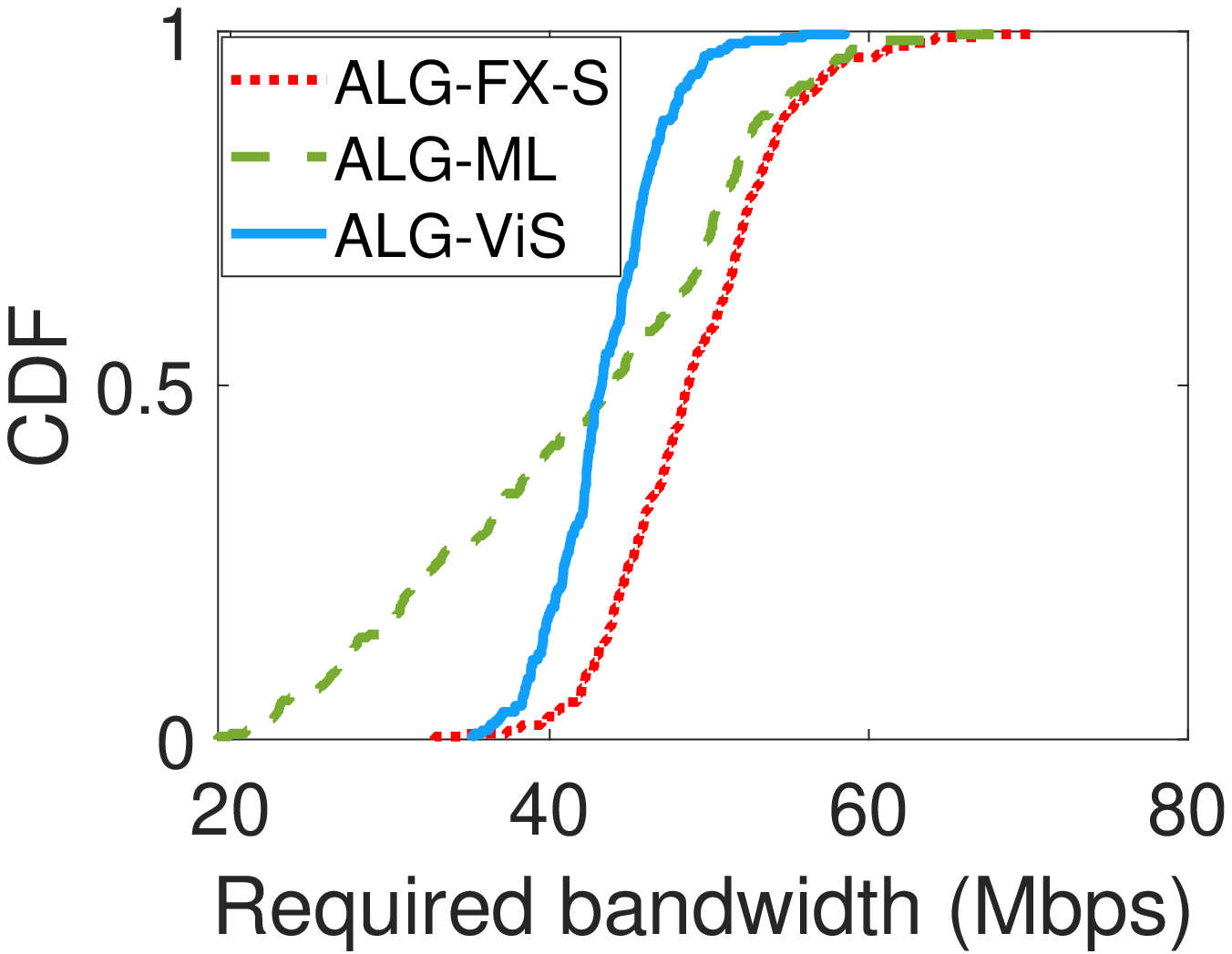}
\caption{CDF of required bandwidth  for 
ALG-ViS and the baselines
 in the edge-assisted  system.
}
\label{fig: bandwidth}
\end{minipage}
\end{figure*}

\begin{figure*}
\centering
\vspace{-0.5cm}
\begin{subfigure}{.24\textwidth}
  \centering
   \includegraphics[width=\linewidth]{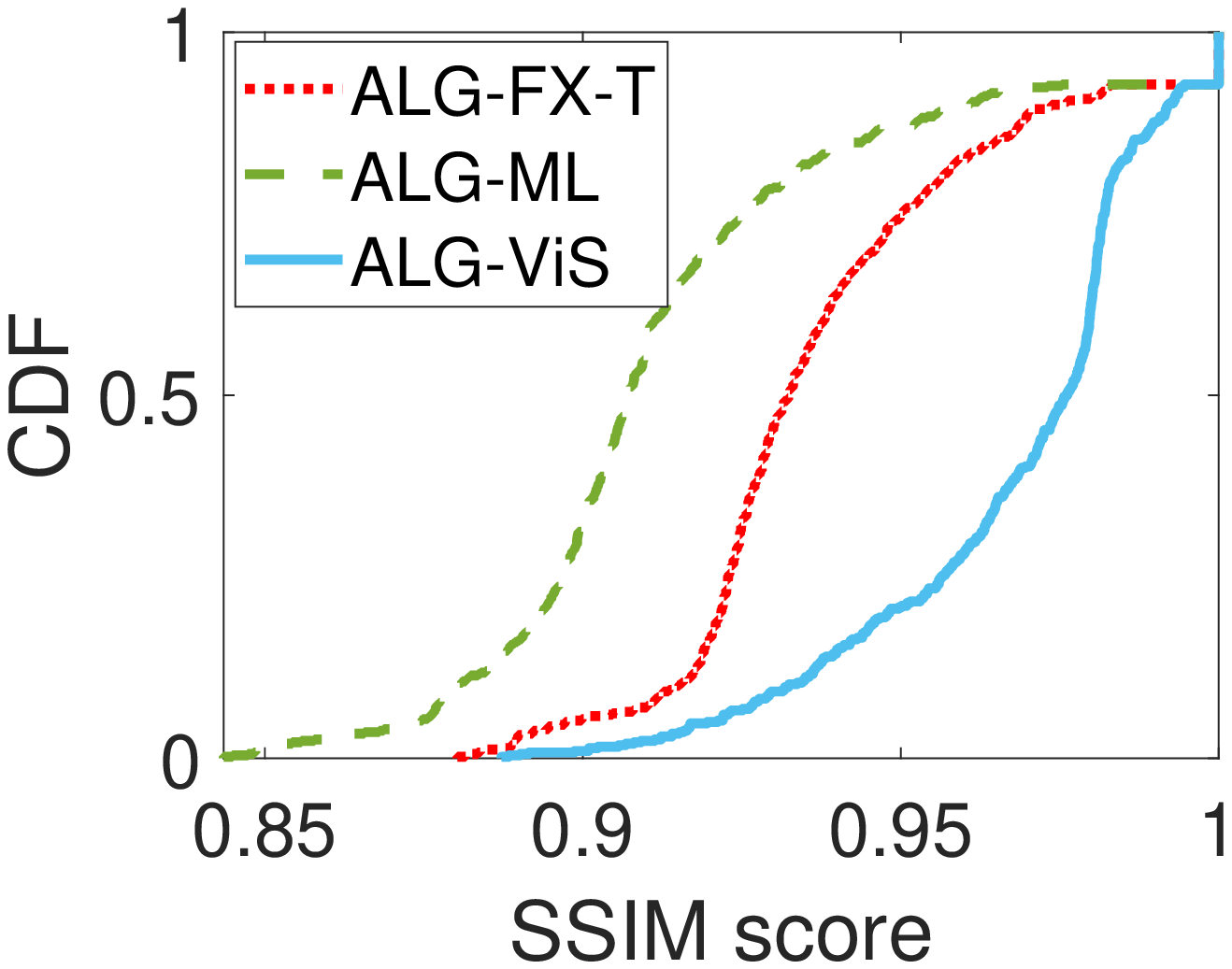}  
  \caption{VR frame quality}
  \label{fig:ssimscore}
\end{subfigure}
\begin{subfigure}{.24\textwidth}
  \centering
   \includegraphics[width=\linewidth]{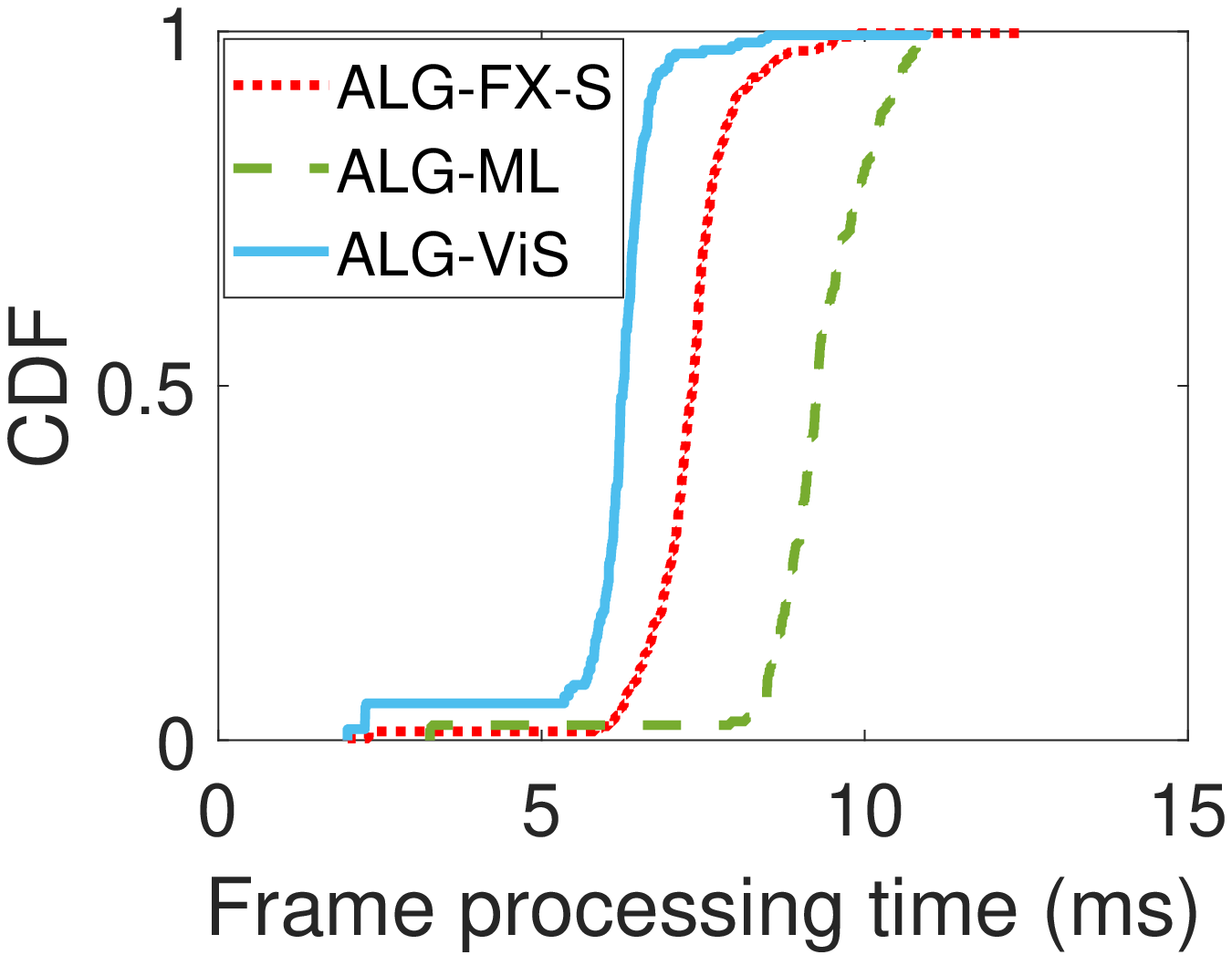}   
  \caption{Frame processing time}
  \label{fig:frame_processing_time}
\end{subfigure}
\begin{subfigure}{.24\textwidth}
  \centering
   \includegraphics[width=\linewidth]{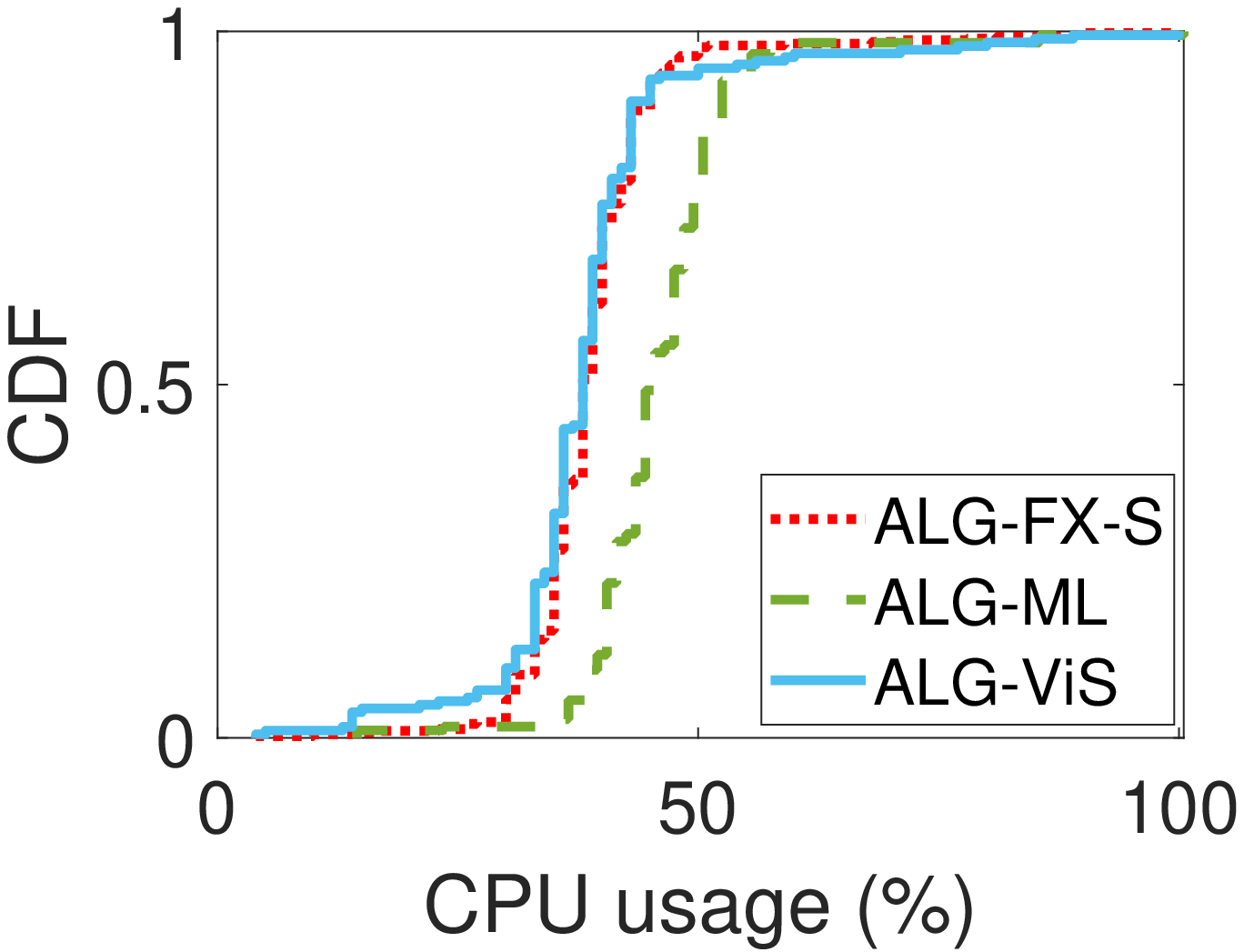}   
  \caption{CPU usage}
  \label{fig:cpu}
\end{subfigure}
\begin{subfigure}{.24\textwidth}
  \centering
   \includegraphics[width=\linewidth]{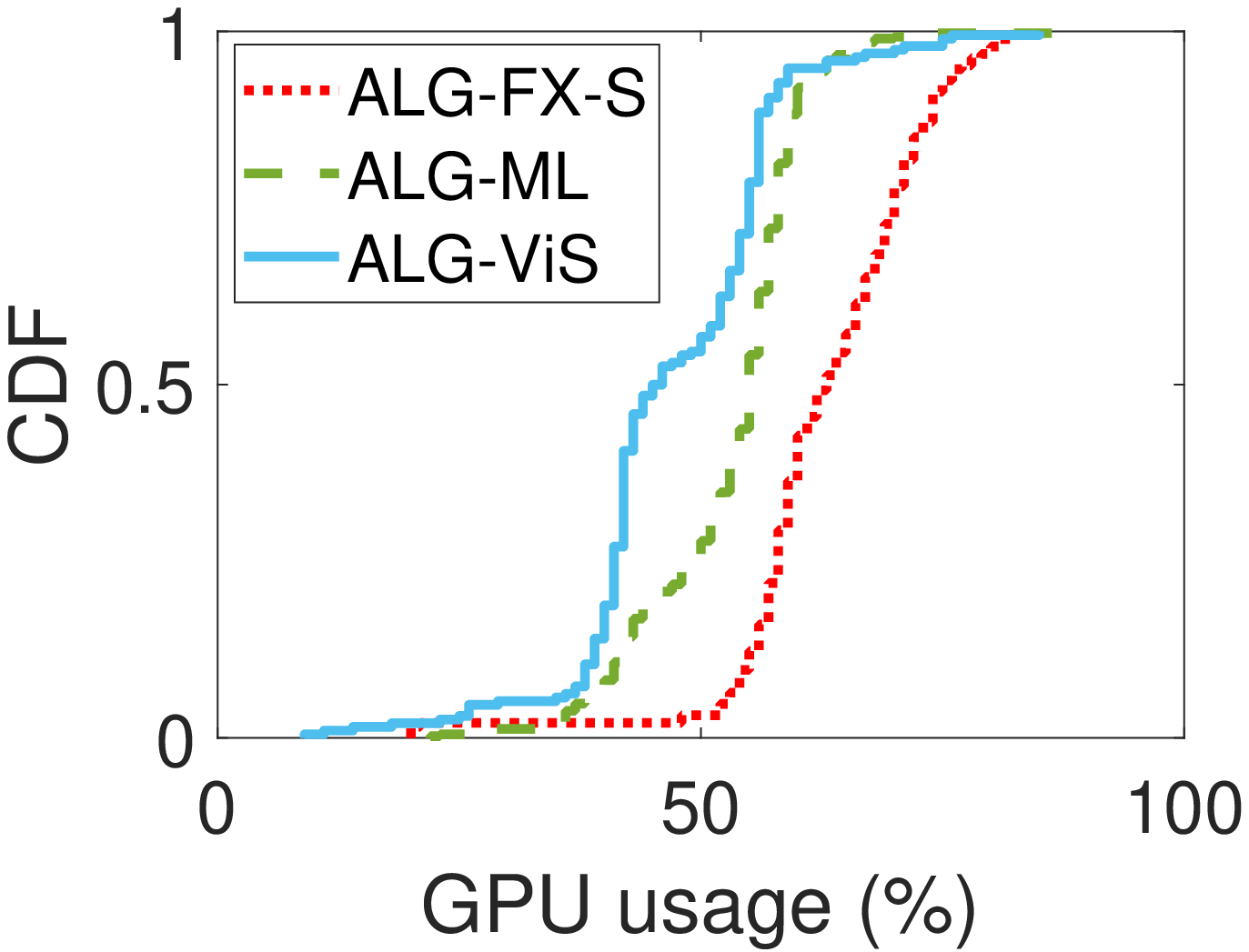}   
  \caption{GPU usage}
  \label{fig:gpu}
\end{subfigure}
  \vspace{-0.16cm}
  \caption{CDFs of 4 different metrics for ALG-ViS and the baselines in the local rendering system. 
  }
  \vspace{-0.65cm}
     \label{ev_vis_alg}
\end{figure*}

We verify the analysis of ViS via simulations in \S\ref{exp} 
and 
examine the performance of ALG-ViS in real-world VR implementations in \S\ref{ev:performance_alg}. The parameters are listed in Table~\ref{tab:para} unless otherwise specified.

\subsection{Examining ViS}
\label{exp}
\subsubsection{Simulation settings}
\begin{table}
\caption{Parameters} 
\vspace{-0.2cm}
\centering
\label{tab:para}
\setlength{\tabcolsep}{1mm}{
\begin{tabular}{cc|cc}
\hline
Parameter & Value& Parameter & Value \\
\hline
VR frame resolution &1080$\times$1080& 
VR frame rate & 60 fps\\
AoV& $90^\circ$&
Mapping scheme &Equiangular\\
Depth encoding & Linear&
Depth map precision & 8 bits\\
Far clipping plane $d_{fp}$& 50 m&
Eye level height &1.6 m\\
Unity unit: meter & 1:1&
Pairs of viewport poses &5000\\
\hline
\end{tabular}}
\vspace{-0.1cm}
\end{table}

We verify the ViS analysis via  simulations using Unity Engine 2019.2.14f1~\cite{Unity20} with 3  VR games listed in Table~\ref{tab:vrgame}. The results are analyzed in MATLAB.  

To simulate the ViS, we randomly sample 5000 pairs of viewport poses from the collected pose trajectories for reference and novel cameras. We obtain
the \textit{pristine novel frame} rendered by Unity, and the \textit{generated  novel frame}  by view projection from the reference frame and its depth map, where  each pixel in the depth map represents the distance of the VR contents to the reference camera. Among the generated novel frame's pixels whose corresponding 3D points are at a distance $d$ from the reference camera, the pixels with RGB values with indistinguishable differences from the pristine novel frame constitute the set $\Phi\left(d\right)$; the other pixels form the set $\Phi^c\left(d\right)$.  The 
ViS given $d$ is calculated as \textit{the proportion  of matched pixels}  ${\widetilde\ViS}\left(d\right)=\frac{{\left| {\Phi \left( d \right)} \right|}}{{\left| {\Phi \left( d \right) \cup {\Phi ^c}\left( d \right)} \right|}}$.
We consider pixel values as indistinguishable when the difference of the pixel values is less than ${C_{th}}\in {\mathbb N}^ +$  in all RGB channels\footnote{$C_{th}$ is selected as the maximum value such that the structural similarity index measure (SSIM) between the combined and the pristine novel frame is larger than 0.95. The combined novel frame is obtained  by replacing the mismatched pixels 
$\mathop  \cup \limits_d {\Phi ^c}\left( d \right)$ in the generated novel frame with the corresponding pixels in the pristine novel frame.  }.
 \subsubsection{Numerical results}

Fig.~\ref{ev_interface} shows the ViS obtained in our simulations, ${\widetilde\ViS}\left(d\right)$, and the analytically derived ViS, $\ViS \left(d\right)$, for headset VR and phone-based VR, for different $d$ and $\Delta t$. The results for desktop VR are similar to the results for phone-based VR and are omitted.  
As expected, the ViS declines with $\Delta t$ (i.e., frames that are separated by a longer time interval are less similar), 
and increases with $d$
(i.e., contents that are farther from the camera change less across different frames). We observe differences between VR interfaces as well: the ViS for headset VR (Fig.~\ref{ev_interface}(a)) is smaller than the ViS for the other interface types (Fig.~\ref{ev_interface}(b)).  
This is explained by the difference in movement dynamics we observed in our dataset: in headset VR, the users change their viewport orientations and flight directions more rapidly than in the other VR interfaces. Smaller ViS for headset VR implies that \emph{fewer VR contents will be classified as background contents}. 
Finally, we note that the gap between the analytical results and the simulations is small, 4.1\% on average for $d<20$m and \emph{only 1.6\% on average for $d\geqslant 20$m}. The gap can be attributed to the omission of object occlusions from our analytical derivations. The achieved highly accurate analytical $\ViS\left(d\right)$ in the high ViS regime (e.g., $d\geqslant 20$m) is crucial for selecting the distance threshold $d_{tr}$ to ensure that the background has a high ViS.

The obtained ViS for the three VR games in desktop VR is shown in Fig.~\ref{ev_game}. The results for the other 2 interface types exhibit the same trends and are omitted. 
Although VK has a larger number of triangles and vertices, which manifests in higher visual scene complexity, the $\ViS\left(d\right)$ of VK is larger than the $\ViS\left(d\right)$ of Office and Lite. 
This is because users' pose trajectories in VK have relatively smaller  $b_{l,\theta}$ and $b_{l}$, and smaller $\mu$ and $\lambda$, corresponding to larger flight lengths and pause durations. In other words, users  tend to change both their viewport orientations and positions more slowly in VK. We hypothesize that higher-complexity games may potentially be more engaging, which encourages the users to explore them in a slower, more deliberate fashion. The observed differences between the ViS 
for different 
games suggest that it is important to take specific game's pose characteristics into account. 
\subsection{Performance of ALG-ViS}
\label{ev:performance_alg}

We implement ALG-ViS and a set of baselines in two VR rendering systems, an on-device (``local") one and one supported by edge computing. Both systems display the generated frames in an Oculus Quest 2 VR headset~\cite{Oculus20}. On-device rendering system is implemented with build 30.0; its target frame rate is set to the default 72 fps. 
The  edge-assisted system generates the VR frames in a Lenovo laptop, with Unity Engine 2019.1.14f1 and Google VR SDK \cite{VRSDK}, and sends them to the headset over IEEE 802.11ac WiFi. The laptop is equipped with an AMD Ryzen 7 4800H CPU and an NVIDIA GTX 1660 Ti GPU. The target frame rate of this system is set to the laptop's default 60 fps. 
To ensure reproducibility, for all algorithms our evaluation is based on replaying 30 min of 
headset pose trajectories we collected (see \S\ref{dataset}).
We examine the performance for all 3 games, and present the results for Lite, which are representative. We examine  the required bandwidth for the edge-supported system in Fig.~\ref{fig: bandwidth}, and the achieved SSIM, frame processing time, 
and CPU and GPU usage (monitored using the OVR Metrics Tool~\cite{OVRMetric}) for the local system in Fig.~\ref{ev_vis_alg}.

We compare ALG-ViS to different approaches with the same 
reference frame interval 
$R$: ALG-FX-S, ALG-FX-T, and ALG-ML. In ALG-ViS, 
$\ViS_{tr}$ is set to 0.945 to ensure high ViS values for background contents. We set $R$ as the minimum integer such that $d_{tr}<d_{fp}$ for every frame. In ALG-FX-S and ALG-FX-T, $d_{tr}$ is determined by the number of triangles $N_{tr}$ in the VR frame, similar to the near and far background splitting in~\cite{MengJiayi20}. Specifically, $d_{th}=\max\left\{qN_{tr},d_{fp}\right\}$, where  $q$ is fixed for each game. 
For a fair comparison, in ALG-FX-T, we set $q$ to make the average frame processing time of ALG-FX-T and ALG-ViS the same to compare the SSIM; in ALG-FX-S, we set $q$ to make the average SSIM  the same as \mbox{ALG-ViS} when comparing other metrics.
In ALG-ML, 
we adopt an online ridge regression model, which has been shown to achieve state-of-the-art accuracy 
in
360$^\circ$ video pose prediction~\cite{kNN,Flare}. 
Following \cite{Firefly},
we predict $x,y,z,\theta$, and $\phi$ separately. We set the history and prediction windows as in \cite{Firefly}. 
We split the VR contents by calculating the ViS of the reference frame and the predicted VR frame.

The CDF of the required bandwidth of the edge-assisted rendering system shown in Fig.~\ref{fig: bandwidth} demonstrates that \mbox{ALG-ViS} requires less bandwidth on average than ALG-FX-S (11.3\% difference; \mbox{ALG-FX-S} consistently transmits more pixels to maintain the same SSIM as ALG-ViS), and has significantly smaller bandwidth variance than ALG-ML (88.4\% reduction). 
Although ALG-ML can save bandwidth when it accurately predicts the viewport pose, the required bandwidth increases drastically when the prediction is erroneous.
Generating less bursty traffic, ALG-ViS prevents transmission resource over-provisioning and potential TCP incast problems in edge-assisted VR systems, which helps supporting these systems better when compared with ALG-ML.

Fig.~\ref{ev_vis_alg} shows the CDFs of the SSIM, frame processing time, and CPU and GPU usage for the local rendering system. It  shows that ALG-ViS improves frame quality and frame processing time while consuming fewer resources.  
The ALG-ViS ensures high average SSIM,
outperforming ALG-FX-T by 3.2\% and ALG-ML by 5.9\%. 
The ALG-ML exhibits lower frame quality. This is because ALG-ML splits the foreground and background contents based on the predicted pose, and prediction errors lead to severe performance degradation~\cite{YangLi2019}. The ALG-ViS decreases the frame processing time by 16.1\% and 33.4\% compared to ALG-FX-S and ALG-ML. The CPU usage of ALG-ML is 20.5\% higher than that of ALG-ViS  due to the extra computation 
required to tune the  regularization parameter and conduct the pose prediction. The GPU usage of ALG-FX-S is 33.7\% higher than that of ALG-ViS because ALG-FX-S classifies more VR contents as foreground contents on average. These results demonstrate that the developed ALG-ViS is lightweight yet effective. 

%% file: Conclusion.tex
In this paper, we first propose a 
viewport pose model for VR systems based on the experimental measurements. We apply the pose model to adaptively select background contents that are reused across VR frames to reduce the communication and computation resource consumption, via quantifying the similarity of pixels across VR frames.
Numerical results verify the pose model and the inter-frame pixel similarity analysis. 
Oculus Quest 2-based implementations of our adaptive background content selection approach 
show that it improves the image quality by 5.6\%
and reduces the variance of the required bandwidth by 88.4\% compared to the method 
based on viewport pose prediction. 

%% file: appendix.tex
\label{proofdisplacement_case1}
\begin{figure*}[h]
\begin{equation}
\begin{aligned}
\label{eqn:large_eqn_displacement}
&{\mathbb{E}\left[ {{\psi^k}\mathds{1}\left( {{A_n}} \right)} \right]}
= \sum\limits_{h = 0}^{n - 2} {n-2 \choose h}{{\left( {1 - c} \right)}^h}{c^{n - 2 - h}}{v^{2k}} \left\{ \mathds{1}\left[ {n = 2} \right] \cdot {\int_0^{\Delta t} {{{\left( {\Delta t - \mathfrak{s}} \right)}^{2k}} }  + \mathds{1}\left[ {n\geqslant 3} \right] \cdot }\int_0^{\Delta t} {\int_0^{\Delta t - \mathfrak{s} }  \cdots  }  \right. \\&
\left. \int_0^{\Delta t - \mathfrak{s} - \sum\limits_{i = j}^{j + n - 3} {T{^\prime_i}} } {\mathbb{E}_{{{\bf{e}}_{j + 1}},{{\bf{e}}_{j + 2}}, \cdots ,{{\bf{e}}_{j + n - 1}}}}\left[ {{{\left( {\sum\limits_{i = j + 1}^{j + n - 2} {T{^\prime_i}{{\bf{e}}_i}}  + \left( {\Delta t -  \mathfrak{s}  - \sum\limits_{i = j + 1}^{j + n - 2} {T{^\prime_i}} } \right){{\bf{e}}_{j + n - 1}}} \right)}^{2k}}} \right] dT{^\prime_{j + 1}} \cdots dT{^\prime_{j + n - 2}}\right\} \\&
\int_0^\mathfrak{s}  {\int_0^{\mathfrak{s} - S{^\prime_j}}  \cdots  } \int_0^{\mathfrak{s} - S{^\prime_j} - \sum\limits_{i = 1}^{h - 1} {S_i^*} } {\int_{\mathfrak{s}  - S{^\prime_j} - \sum\limits_{i = 1}^h {S_i^*} }^\infty  {{\mu ^{n - 1}}} } {e^{ - \mu \left( {\Delta t - \mathfrak{s}} \right)}}{\lambda ^{h + 2}}{e^{ - \lambda \left( {\sum\limits_{i = 1}^h {S_i^*}  + S{^\prime_j} + S{^\prime_{j + n-1}}} \right)}}dS{^\prime_j}dS_1^* \cdots dS_h^*dS{^\prime_{j + n-1}} \cdot  d \mathfrak{s}
\\ =& \sum\limits_{h = 0}^{n - 2} {{n - 2 \choose h }{{\left( {1 - c} \right)}^h}{c^{n - 2 - h}}{v^{2k}}\int_0^{\Delta t} {{\mu ^{n - 1}}{e^{ - \mu \Delta t}}{\lambda ^{h + 1}}{n + k - 2 \choose n - 2}\frac{{{{\left( {\Delta t - \mathfrak{s}} \right)}^{2k + n - 2}}\left( {2k} \right)!}}{{\left( {2k + n - 2} \right)!}}{e^{ - \left( {\lambda  - \mu } \right)\mathfrak{s}}}\frac{{{\mathfrak{s}^{h + 1}}}}{{\left( {h + 1} \right)!}}d\mathfrak{s}} } \\
= &{n+k-2 \choose n-2} \sum\limits_{h = 0}^{n - 2}  {{\mu }{\lambda }{g_{n - 2,h,h + 2,k,2k + n + h + 1}}}
\end{aligned}
\end{equation}
\vspace{-0.4cm}
\hrulefill
\vspace{-0.1cm}
\end{figure*}

\subsection{Proof of Lemma~\ref{displacement_case1}}

When Case 1 holds, there are two possible ways to end the observation interval at time $t_s + \Delta t$. If $\exists j^\prime \in \mathbb{N}^{+}$, $\sum\limits_{n = 0}^{j^\prime} {{T_n} + \sum\limits_{n = 0}^{j^\prime - 1} {{S_n}} } <t_s+ \Delta t< \sum\limits_{n = 0}^{j^\prime} {{T_n} + \sum\limits_{n = 0}^{j^\prime} {{S_n}} } $, the observation end time $t_s + \Delta t$  lies in a pause interval.
Otherwise, $t_s + \Delta t$  lies in a flight interval. If it lies in a pause interval, the movement includes some number of complete flights without any fractional flights.  We do not observe any movement when $n=1$. For $n \geqslant 1$, $A_n$ is the event that $t_s + \Delta t$ lies in a pause interval and there are $n-1$ complete flights in $[t_s, t_s + \Delta t]$. This leads to the first summation in \eqref{Eqn: displacement1}. 

If $t_s + \Delta t$ lies in a flight interval, the movement includes some number of complete flights and a fraction of the last flight.  $B_n$ is the event that $t_s + \Delta t$ lies in a flight interval and there are $n-1$ complete flights in $[t_s, t_s + \Delta t]$.  This leads to the second summation in \eqref{Eqn: displacement1}.

It only remains to calculate ${\mathbb{E}\left[ {{\psi^k}\mathds{1}\left( {{A_n}} \right)} \right]}$ and ${\mathbb{E}\left[ {{\psi^k}\mathds{1}\left( {{B_n}} \right)} \right]}$  to conclude the proof. 
We start with ${\mathbb{E}\left[ {{\psi^k}\mathds{1}\left( {{A_n}} \right)} \right]}$.  Given $T_i ,{\bf e}_{i}$, the $k$-th moment of $\psi$ is given by ${v^{2k}}{\left\| {\sum\limits_{i = j+1}^{j+n - 1} {{T_i}{{\bf e}_i}} } \right\|^{2k}}$.
According to the flight direction model in \S\ref{modifiedrwp}, the direction of ${\bf e}_{i}$ is i.i.d. uniformly distributed on $\left[0,2\pi\right)$. Thus,  ${\mathbb{E}_{{{\bf e}_i},{{\bf e}_k}}}\left[{\bf e}_{i}{\bf e}_{k}\right]$ is equal to $1$ for $i=k$ and  $0$ for $i\ne k$. Thus, expanding ${\left\| {\sum\limits_{i = j+1}^{j+n - 1} {{T_i}{{\bf e}_i}} } \right\|^{2k}}$ and averaging over ${{{\bf e}_{j+1}},{{\bf e}_{j+2}}, \cdots ,{{\bf e}_{j+n - 1}}}$,  only the terms in which all $T_i{\bf e}_i$ ($j+1\leqslant i\leqslant j+n-1$) have even powers are non-zero. The number of these non-zero terms is ${n+k-2 \choose n-2}$. Further, let $h \left(h\leqslant n-1\right)$ denote the number of non-zero complete pause intervals (excluding the first and last fractional pause intervals). Denote the set of the non-zero complete pause intervals as $\{ S_i^*\} ,1 \leqslant i \leqslant h$. Let $\mathfrak{s} = \Delta t - \sum\limits_{i = j+1}^{j+n - 1} {{T_i}} $ be the sum of pause intervals in the observation interval.  ${\mathbb{E}\left[ {{\psi^k}\mathds{1}\left( {{A_n}} \right)} \right]}$ is given by \eqref{eqn:large_eqn_displacement}, shown at the top of next page. Using these observations, the expressions for ${\mathbb{E}\left[ {{\psi^k}\mathds{1}\left( {{B_n}} \right)} \right]}$ can be obtained similarly, which completes the proof.

\subsection{Proof of Lemma \ref{Lemma: Probability of Starting During Flight}} \label{Appendix: Probability of Starting Druing a Flight}

Let $\Gamma_n$ be the end time of the $n$-th flight, i.e., $\Gamma_n = \sum_{i=0}^{n-1} S_i +\sum_{i=0}^{n} T_i$. Then, for any $T > 0$, there exists $n \in \mathbb{N}$ such that $\Gamma_n \leqslant T < \Gamma_{n+1}$, and we can upper and lower bound $p_T$ according to 
\begin{equation*}
 \mathbb{E}\left[ {\frac{{\sum\limits_{i = 0}^n {{T_i}} }}{{\sum\limits_{i = 0}^n {{S_i}}  + \sum\limits_{i = 0}^n {{T_i}} }}} \right] \leqslant p_T \leqslant \mathbb{E}\left[ {\frac{{\sum\limits_{i = 0}^{n + 1} {{T_i}} }}{{\sum\limits_{i = 0}^{n - 1} {{S_i}}  + \sum\limits_{i = 0}^{n + 1} {{T_i}} }}} \right].
\end{equation*}

The lower bound includes an extra pause interval between $\Gamma_n$ and $\Gamma_{n+1}$ by ignoring any possible fractional flight duration. On the other hand, the upper bound includes one complete flight duration between $\Gamma_n$ and $\Gamma_{n+1}$ by ignoring $S_n$. 

When $T\to\infty$, using Lebesgue's dominated convergence theorem to replace the order of limit and expectation operators and the law of large numbers (as $T$ grows large, $n$ tends to infinity), it can be seen that both upper and lower bounds converge to $p = \frac{{\lambda /\left( {1 - c} \right)}}{{\lambda /\left( {1 - c} \right) + \mu }}$.  

\subsection{Proof of Theorem~\ref{prop:fovterm}}
\label{proof_fov}

Seen from Fig.~\ref{fig:fovterm}, the overlapping azimuth angle is
\begin{equation}
\label{eq:pf_fov}
\begin{aligned}
&{\mathbb{E}_{\Delta \phi }}\left[ {\frac{{\min \left( {{\phi _{\rm ref}},{\phi _{\rm nov}}} \right) + \frac{{{w _{fv}}}}{2} - \left( {\max \left( {{\phi _{\rm ref}},{\phi _{\rm nov}}} \right) - \frac{{{w _{fv}}}}{2}} \right)}}{{{w _{fv}}}}} \right]\\
\labrel={fovexplain1} & 1 - \frac{1}{{{w _{fv}}}}\int_{ - {w _{fv}}}^{{w _{fv}}} {\left| {\Delta \phi } \right|{p_{\Delta \phi }}\left( {\Delta \phi } \right)d} \left( {\Delta \phi } \right)\\
=&1 - \frac{2}{{{w _{fv}}}}\int_0^{{w _{fv}}} {\Delta \phi {p_{\Delta \phi }}\left( {\Delta \phi } \right)d} \left( {\Delta \phi } \right)=p^\phi_f.
 \end{aligned}
\end{equation}
where ${\Delta \phi }=\phi_{\rm ref}-\phi_{\rm nov}$, and \eqref{fovexplain1} is because 
$\min \left( {{\phi _{{\rm{ref}}}},{\phi _{{\rm{nov}}}}} \right) + \frac{{{\omega _{fv}}}}{2} - \left( {\max \left( {{\phi _{{\rm{ref}}}},{\phi _{{\rm{nov}}}}} \right) - \frac{{{\omega _{fv}}}}{2}} \right) = {\phi _{{\rm{ref}}}} + \frac{{{\omega _{fv}}}}{2} - \left( {{\phi _{{\rm{nov}}}} - \frac{{{\omega _{fv}}}}{2}} \right)$ if $\phi _{\rm ref}<\phi _{\rm nov}$ and $\min \left( {{\phi _{{\rm{ref}}}},{\phi _{{\rm{nov}}}}} \right) + \frac{{{\omega _{fv}}}}{2} - \left( {\max \left( {{\phi _{{\rm{ref}}}},{\phi _{{\rm{nov}}}}} \right) - \frac{{{\omega _{fv}}}}{2}} \right) = {\phi _{\rm nov}} + \frac{{{\omega _{fv}}}}{2} - \left( {{\phi _{\rm ref}} - \frac{{{\omega _{fv}}}}{2}} \right)$ if $\phi _{\rm ref}\geqslant\phi _{\rm nov}$. ${p_{\Delta \phi }}\left(\Delta \phi\right)$ is the PDF of $\Delta\phi$ obtained in Section~\ref{sub:orientation}:
\begin{align*}
&{p_{\Delta \phi }}\left( {\Delta \phi } \right) \\ = &
\left\{ {\begin{array}{*{20}{l}}
{\frac{1}{{2{b_l}}}\exp \left( { - \frac{{\left| {\Delta \phi } \right|}}{{{b_l}}}} \right),\;\Delta t < {\beta _1}}\\
{\left( {1 - {p_l}} \right)\frac{{\exp \left( { - \frac{{\left| {\Delta \phi } \right| - {\mu _{lo}}}}{{{b_{lo}}}}} \right)}}{{{b_{lo}}{{\left( {1 + \exp \left( { - \frac{{\left| {\Delta \phi } \right| - {\mu _{lo}}}}{{{b_{lo}}}}} \right)} \right)}^2}}}\frac{1}{{2{{\left( {1 + \exp \left( { - \frac{\pi }{{{b_{lo}}}}} \right)} \right)}^{ - 1}} - 1}}}\\
{ + {p_l}\frac{{\frac{1}{{2{b_l}}}\exp \left( { - \frac{{\left| {\Delta \phi } \right|}}{{{b_l}}}} \right)}}{{{{\left( {1 + \exp \left( { - \frac{\pi }{{{b_{lo}}}}} \right)} \right)}^{ - 1}} - 1}},\;{\beta _1}\leqslant \Delta t < {\beta _2}}\\
{\frac{w_{fv}}{{2\pi }},\;\Delta t > {\beta _2}.}
\end{array}} \right.\end{align*}
Substituting ${p_{\Delta \phi }}\left(\Delta \phi\right)$ into \eqref{eq:pf_fov}, we obtain the results for $p^\phi_f$ (the fraction of overlapping azimuth angles) in \eqref{eq:fov2}.

Similarly, we have
\begin{equation}
\label{eq:overlapping_polar}
\begin{aligned}
&{\mathbb{E}_{\Delta \theta }}\left[ {\frac{{\min \left( {{\theta _{\rm ref}},{\theta _{\rm nov}}} \right) + \frac{{{w _{fv}}}}{2} - \left( {\max \left( {{\theta _{\rm ref}},{\theta _{\rm nov}}} \right) - \frac{{{w _{fv}}}}{2}} \right)}}{{{w _{fv}}}}} \right]\\
 =&1 - \frac{2}{{{w _{fv}}}}\int_0^{{w _{fv}}} {\Delta \theta {p_{\Delta \phi }}\left( {\Delta \theta } \right)d} \left( {\Delta \theta } \right)\\
 =&\left( 1 - \frac{{{b_{l,\theta }} - \exp \left( { - \frac{{{w _{fv}}}}{{{b_{l,\theta }}}}} \right)\left( {{b_{l,\theta }} + {w _{fv}}} \right)}}{{{w _{fv}}}}\right)
\end{aligned}
\end{equation}
where ${p_{\Delta \theta }}\left( {\Delta \theta } \right) = \frac{1}{{2{b_{l,\theta }}}}\exp \left( { - \frac{{\left| {\Delta \phi } \right|}}{{{b_{l,\theta }}}}} \right)$ is the PDF of $\Delta\theta$ obtained in \S\ref{sub:orientation}. 

Combining fractions of overlapping azimuth angle and overlapping polar angle in \eqref{eq:fov2}  and \eqref{eq:overlapping_polar}, we get the final result of $\ViS_{fov}$.

\subsection{Proof of Theorem~\ref{distance_term}}
\label{p_d_proof}
$\vartheta$ is uniformly distributed on $\left( { - \frac{{{w_{fv}}}}{2},\frac{{{w_{fv}}}}{2}} \right)$. This is because the walking direction is the same as $\phi_{\rm{ref}}$ (see \S\ref{sec:correlation}), and the included angle of $\overrightarrow {{{\bf{X}}_{{\rm{ref}}}}{{\bf{X}}_{{\rm{3D}}}}} $ and positive direction of $X$-axis is uniformly distributed on $\left( {{\phi _\text{ref}} - \frac{{{w_{fv}}}}{2},{\phi _\text{ref}} + \frac{{{w_{fv}}}}{2}} \right)$.
 Taking the average over $\vartheta$ and substituting $m\left(1\right)$, \eqref{def:p_d_def} is rewritten as $\ViS_{dst}\left(d\right)=1 + \frac{m\left(1\right)}{{{d^2}}} - \frac{{4\sin \left( {\frac{{{w_{fv}}}}{2}} \right)}}{{{w_{fv}}d}}{\mathbb{E}_\psi }\left[ \sqrt{\psi}  \right]$, where
$\mathbb{E}_\psi \left[ \sqrt{\psi}  \right]$ is calculated as ${\mathbb{E}_\psi }\left[ \sqrt{\psi}  \right] = \frac{1}{{\sqrt \pi  }}\mathbb{E}_{\psi}\left[ {\int_0^\infty  {\psi{e^{ - \tau \psi}}{\tau ^{ - \frac{1}{2}}}} d\tau } \right]$.  We further approximate $\mathbb{E}_\psi \left[\sqrt \psi \right]$ by $\frac{1}{{\sqrt \pi  }}\mathbb{E}_{\psi}\left[ {\int_0^{\frac{1}{{{\varepsilon ^2}}}} {\psi {e^{ - \tau \psi }}{\tau ^{ - \frac{1}{2}}}} d\tau } \right]$ with the error  $\frac{1}{{\sqrt \pi  }}\mathbb{E}_{\psi}\left[ {\int_{\frac{1}{{{\varepsilon ^2}}}}^\infty  {\psi{e^{ - \tau \psi}}{\tau ^{ - \frac{1}{2}}}} d\tau } \right]$. $\mathbb{E}_{\psi}\left[ {\int_0^{\frac{1}{{{\varepsilon ^2}}}} {\psi {e^{ - \tau \psi }}{\tau ^{ - \frac{1}{2}}}} d\tau } \right]$ is given as
\begin{align*}
&\mathbb{E}_{\psi}\left[ {\int_0^{\frac{1}{{{\varepsilon ^2}}}} {\psi {e^{ - \tau \psi }}{\tau ^{ - \frac{1}{2}}}} d\tau } \right]\\
 =& \sum\limits_{i = 0}^\infty  {\left\{ {\mathbb{E}_{\psi}\left[ {\int_0^{\frac{1}{{{\varepsilon ^2}}}} {\psi\frac{{{{\left( { - \tau \psi} \right)}^i}}}{{i!}}{\tau ^{ - \frac{1}{2}}}} d\tau } \right]} \right\}} 
 = \sum\limits_{i = 0}^\infty  { g\left(i\right)} .
\end{align*}
The approximation error  $\mathbb{E}_{\psi}\left[ {\int_0^{\frac{1}{{{\varepsilon ^2}}}} {\psi {e^{ - \tau \psi }}{\tau ^{ - \frac{1}{2}}}} d\tau } \right]$ is smaller than $\frac{\varepsilon }{{\sqrt \pi  }} \mathcal{M}_{\psi}\left( { - \frac{1}{{{\varepsilon ^2}}}} \right)$, i.e., we have
\begin{equation*}
\begin{aligned}
0 &< \mathbb{E}_{\psi}\left[ {\int_{\frac{1}{{{\varepsilon ^2}}}}^\infty  {\psi{e^{ - \tau \psi}}{\tau ^{ - \frac{1}{2}}}} d\tau } \right] < \varepsilon\mathbb{E}_{\psi} \left[\int_{\frac{1}{{{\varepsilon ^2}}}}^\infty  {\psi {e^{ - \tau \psi }}d\tau }  \right]\\& = \varepsilon {\mathbb{E}_\psi }\left[ {{e^{ - \frac{\psi }{{{\varepsilon ^2}}}}}} \right] = \varepsilon {\mathcal{M}_\psi }\left( { - \frac{1}{{{\varepsilon ^2}}}} \right)<\varepsilon
\end{aligned}
\end{equation*}
which concludes the proof.